\documentclass[leqno,12pt,a4wide]{article}
\usepackage{a4wide}
\usepackage[]{amsfonts} \usepackage[]{amsmath}
\usepackage[]{amssymb} \usepackage[]{latexsym}
\usepackage{graphicx,color} \usepackage{amsthm}
\usepackage{mathrsfs} \usepackage{url}
\usepackage{cancel} \usepackage{mathtools}
\usepackage[sort]{natbib}

\usepackage{xspace,color}
\usepackage[normalem]{ulem}
\numberwithin{equation}{section}

\mathtoolsset{showonlyrefs}

\begin{document}

\newtheorem{teo}{Theorem}[section] \newtheorem*{teo*}{Theorem}
\newtheorem{prop}[teo]{Proposition} \newtheorem*{prop*}{Proposition}
\newtheorem{lema}[teo]{Lemma} \newtheorem*{lema*}{Lemma}
\newtheorem{cor}[teo]{Corollary} \newtheorem*{cor*}{Corollary}

\theoremstyle{definition}
\newtheorem{defi}[teo]{Definition} \newtheorem*{defi*}{Definition}
\newtheorem{exem}[teo]{Example} \newtheorem*{exem*}{Example}
\newtheorem{obs}[teo]{Remark} \newtheorem*{obs*}{Remark}
\newtheorem*{hipo}{Hypotheses}
\newtheorem*{nota}{Notation}

\newcommand{\ds}{\displaystyle} \newcommand{\nl}{\newline}
\newcommand{\eps}{\varepsilon}
\newcommand{\LMMR}{\mbox{LMMR}}
\newcommand{\bE}{\mathbb{E}}
\newcommand{\cB}{\mathcal{B}}
\newcommand{\cF}{\mathcal{F}}
\newcommand{\cA}{\mathcal{A}}
\newcommand{\cM}{\mathcal{M}}
\newcommand{\cD}{\mathcal{D}}
\newcommand{\cN}{\mathcal{N}}
\newcommand{\cL}{\mathcal{L}}
\newcommand{\cLN}{\mathcal{LN}}
\newcommand{\bP}{\mathbb{P}}
\newcommand{\bQ}{\mathbb{Q}}
\newcommand{\bN}{\mathbb{N}}
\newcommand{\bR}{\mathbb{R}}
\newcommand{\rhor}{\raisebox{1.5pt}{$\rho$}}
\newcommand{\varphir}{\raisebox{1.5pt}{$\varphi$}}
\newcommand{\taur}{\raisebox{1pt}{$\tau$}}
\newcommand{\VIX}{\mbox{VIX}}

\title{Multiscale Stochastic Volatility Model for Derivatives on Futures}

\author{Jean-Pierre Fouque\thanks{Department of Statistics \& Applied Probability, University of California, Santa Barbara, CA 93106-3110,
{\em fouque@pstat.ucsb.edu}. Work supported by NSF grant DMS-1107468.} ,
Yuri F. Saporito\thanks{Department of Statistics \& Applied Probability, University of California, Santa Barbara, CA 93106-3110,
{\em saporito@pstat.ucsb.edu}. Work supported by Fulbright grant 15101796 and
The CAPES Foundation, Ministry of Education of Brazil, Bras\'ilia, DF 70.040-020, Brazil.} ,
Jorge P. Zubelli\thanks{IMPA (Instituto de Matem\'atica Pura e Aplicada), Est. D. Castorina 110, Rio de Janeiro, RJ 22460-320, Brazil, {\em zubelli@impa.br}. Work supported by CNPq under grants 302161 and 474085, and by FAPERJ under the CEST and PENSARIO programs.}}

\maketitle

\abstract{In this paper we present a new method to compute the first-order approximation of the price of derivatives on futures in the context of multiscale stochastic volatility of Fouque \textit{et al.} (2011, CUP). It provides an alternate method to the singular perturbation technique presented in Hikspoors and Jaimungal (2008). The main features of our method are twofold: firstly, it does not rely on any additional hypothesis on the regularity of the payoff function, and secondly, it allows an effective and straightforward calibration procedure of the model to implied volatilities. These features were not achieved in previous works. Moreover, the central argument of our method could be applied to interest rate derivatives and compound derivatives. The only pre-requisite of our approach is the first-order approximation of the underlying derivative. Furthermore, the model proposed here is well-suited for commodities since it incorporates mean reversion of the spot price and multiscale stochastic volatility. Indeed, the model was validated by calibrating it to options on crude-oil futures, and it displays a very good fit of the implied volatility.}


\section{Introduction}

In many financial applications the underlying asset of the derivative contract under consideration is a derivative itself. A very important example of this complex and widely traded class of products consists of derivatives on future contracts. We shall study such financial instruments in the context of multiscale stochastic volatility as presented in \cite{fouque_stochastic_vol_new}.

It is well-known that under the No-Arbitrage Hypothesis, one can find a risk-neutral probability measure such that all \textit{tradable} assets in this market, when properly discounted, are martingales under this measure (see \cite{math_of_arbitrage} for an extensive exposition on this subject). Here, we assume constant interest rate throughout this paper.

The future contract on the asset $V$ with maturity $T$ is a standardized contract traded at a futures exchange for which both parties consent to trade the asset $V$ at time $T$ for a price agreed upon the day the contract was written. This previously arranged price is called \textit{strike}. The \textit{future price} at time $t$ with maturity $T \geq t$ of the asset $V$, which will be denoted by $F_{t,T}$, is defined as the strike of the future contract on $V$ with maturity $T$ such that no premium is paid at time $t$. In symbols,
\begin{align}\label{eq:fut_price_intro}
F_{t,T} = \bE_{\bQ}[V_T \ | \ \cF_t],
\end{align}
where $\bQ$ is a risk-neutral probability. If the asset $V$ is tradable, then we simply have $F_{t,T} = e^{r(T-t)}V_t$, where $r$ is the constant interest rate, and then derivatives on futures can be treated in the exact same way one handles derivatives on the asset itself.

When interest rate is constant, future prices are non-trivial when the asset is not tradable and therefore the discounted asset price is not a martingale, see for example \cite[Chapter 3]{musirutko08}. This will be our main assumption: \textit{the asset $V$ is not tradable}. More precisely, we assume the asset price presents mean reversion. Some examples of such assets are: commodities, currency exchange rates, volatility indices, and interest rates.

Empirical evidence of the presence of stochastic factors in the volatility of financial assets is greatly documented in the literature, see for example \cite{gatheral06} and references therein. The presence of a fast time-scale in the volatility in the S\&P 500 was reported in \cite{fouque_empirical}. We refer the reader to \cite{fouque_stochastic_vol_new} for comprehensive exposition on this subject. Multiscale stochastic volatility models lead to a first-order approximation of derivatives prices. This approximation is composed by the leading-order term given by the Black-Scholes price with the averaged effective volatility and the first-order correction only involves Greeks of this leading-term. In terms of implied volatility, this perturbation analysis translates into an affine approximation in the log-moneyness to maturity ratio (LMMR). Subsequently, this leads to a simple calibration procedure of the group market parameters, that are also used to compute the first-order  approximation of the price of exotic derivatives.

Because of the nature of our problem, the future price, which is the underlying asset of the derivative in consideration, has its dynamics explicitly depending on the time-scales of the volatility. This creates an important difference from the usual perturbation theory to the derivative pricing problem.

The method presented in this paper can be described as follows:
\begin{itemize}

\item[(i)] Write the stochastic differential equation (SDE) for the future $F_{t,T}$ with all coefficients depending only on $F_{t,T}$. This means we will need to invert the future prices of $V$ in order to write $V_t$ as a function of $F_{t,T}$.

\item[(ii)] Consider the pricing partial differential equation (PDE) for a European derivative on $F_{t,T}$. The coefficients of this PDE will depend on the time-scales of the stochastic volatility of the asset in a complicated way. At this point, we use perturbation analysis to treat such PDE by expanding the coefficients.

\item[(iii)] Determine the first-order approximation of derivatives on $F_{t,T}$ as it is done in \cite{fouque_stochastic_vol_new}.

\end{itemize}

Indeed, this method is not the only way to tackle this problem. Instead, we could have considered this compound derivative as a more elaborate derivative in the asset and then find the first-order approximation proposed in \cite{jaimungal_futures}. This, in turn, follows the idea designed in \cite{multiscale_option_interest_rate} and is based on the Taylor expansion of the payoff under consideration around the zero-order term of the approximation of the future price $F_{t,T}$. Therefore, some smoothness of the payoff function must be assumed. Since the method considered here does not rely on such Taylor expansion, no restriction other than the ones intrinsic to the perturbation method is required. Furthermore, we shall show that although the method presented here is more involved, it allows a cleaner calibration. This is due to the fact we are considering the derivative as a function of the future price, which is a tradable asset and hence a martingale under the pricing risk-neutral measure. We refer to Section \ref{sec:comparison} for a more thoroughly comparison between our method introduced here and the method presented in \cite{jaimungal_futures}.

Another important set of examples that can be handled using the method proposed in this paper consists of interest rate derivatives (see \cite{multiscale_option_interest_rate}). Moreover, in the equity case, the method could be use to tackle the general problem of pricing compound derivatives, as it is done in \cite{multiscale_compound} by Taylor expansion of the payoff function.


The main contribution of our work is a general method to compute the first-order approximation of the price of general compound derivatives such that no additional hypothesis on the regularity of the payoff function must be assumed. The only pre-requisite is the first-order approximation of the underlying derivative. In other words, the method proposed here allows us to derive the first-order approximation of compound derivatives keeping the hypotheses of the original approximation given in \cite{fouque_stochastic_vol_new}. Furthermore, this method maintains another desirable feature of the perturbation method: the direct calibration of the market group parameters.

This paper is organized as follows: Section \ref{sec:model} describes the dynamics of the underlying asset and then, in Section \ref{sec:option} we follow the method previously outlined to find the first-order approximation of derivatives on future contracts of $V$. Section \ref{sec:calibration} characterizes the calibration procedure to call options and we analyze an example of calibration to options on crude-oil futures. 
Finally, we conclude the paper in Sections \ref{sec:comparison} and \ref{sec:conclusion} with a comparison of our work with a previous method and some suggestions for further research.

\section{The Model}\label{sec:model}

Firstly, we fix a filtered risk-neutral probability space $(\Omega, \cF, (\cF_t)_{t \geq 0}, \bQ)$. The risk-neutral measure is chosen so that the relation \eqref{eq:fut_price_intro} holds. In this probability space, we assume that the asset value $V_t$ is described by an exponential Ornstein-Uhlenbeck (exp-OU) stochastic process with a multiscale stochastic volatility. Namely,
\begin{align}\label{eq:sde_risk_neutral}
\left\{
\begin{array}{l}
 V_t = e^{s(t) + U_t}, \\ \\
 dU_t = \kappa(m - U_t)dt + \eta(Y^{\eps}_t,Z^{\delta}_t) dW_t^{(0)}, \\ \\
 \ds dY^{\eps}_t = \frac{1}{\eps} \alpha(Y_t^{\eps})dt + \frac{1}{\sqrt{\eps}} \beta(Y_t^{\eps}) dW_t^{(1)}, \\ \\
 dZ^{\delta}_t = \delta c(Z^{\delta}_t)dt + \sqrt{\delta} g(Z^{\delta}_t) dW_t^{(2)}, \\
\end{array}
\right.
\end{align}
where $(W_t^{(0)},W_t^{(1)},W_t^{(2)})$ is a correlated $\bQ$-Brownian motion with
$$dW_t^{(0)}dW_t^{(i)} = \rhor_i dt, \ i=1,2, \ dW_t^{(1)}dW_t^{(2)} = \rhor_{12} dt.$$
We shall denote by $Y^{1}$ the process given by the second of the stochastic differential equations in \eqref{eq:sde_risk_neutral} when $\eps=1$.

The main assumptions of this model are:
\begin{enumerate}

\item[$\bullet$] There exists a unique solution of the SDE \eqref{eq:sde_risk_neutral} for any fixed $(\eps,\delta)$.

\item[$\bullet$] The risk-neutral probability $\bQ$ is chosen in order to match the future prices of $V$ observed in the market to the prices produced by the model \eqref{eq:sde_risk_neutral} and the martingale relation \eqref{eq:fut_price_intro}.

\item[$\bullet$] $|\rhor_1| < 1$, $|\rhor_2| < 1$, $|\rhor_{12}| < 1$ and $1 + 2\rhor_1 \rhor_2 \rhor_{12} - \rhor_1^2 - \rhor_2^2 - \rhor_{12}^2 > 0$. These conditions ensure the positive definiteness of the covariance matrix of $(W_t^{(0)},W_t^{(1)},W_t^{(2)})$.

\item[$\bullet$] The interest rate is constant and equals $r$.

\item[$\bullet$] $\alpha$ and $\beta$ are such that the process $Y^1$ has a unique invariant distribution and is mean-reverting as in \cite[Section 3.2]{fouque_stochastic_vol_new}.

\item[$\bullet$] $\eta(y,z)$ is a positive function, smooth in $z$ and such that $\eta^2(\cdot,z)$ is integrable with respect to the invariant distribution of $Y^1$.

\item[$\bullet$] $s(t)$ is a deterministic seasonality factor.

\end{enumerate}

It is important to notice that we could have explicitly considered the market prices of volatility risk as it is done in \cite{fouque_stochastic_vol_new} and in doing so we would have a term of order $\eps^{-1/2}$ and a term of order $\delta^{1/2}$ in the drifts of $Y^{\eps}$ and $Z^{\delta}$ respectively, both depending on $Y^{\eps}$ and $Z^{\delta}$, and they could have been handled in the way it is done in the aforesaid reference. For simplicity of notation, we do not consider these market prices of volatility risk here.

A simple generalization of this model is the addition of a deterministic time-varying long run mean $m(t)$ in the drift of $U$, which can be easily handled. A more subtle extension would be to the Schwartz two-factor model, see \cite{schwartz1997}.

We now restate the definition of the future prices of $V$
\begin{align}\label{eq:fut_price}
F_{t,T} = \bE_{\bQ}[V_T \ | \ \cF_t], \ 0 \leq t \leq T,
\end{align}
and then in next the section we will develop the first-order approximation of derivatives on $F_{t,T}$.
\begin{obs}
More precisely, we say that a function $g^{\eps,\delta}$ is a first-order approximation to the function $f^{\eps,\delta}$ if
$$|g^{\eps,\delta} - f^{\eps,\delta}| \leq C(\eps + \delta),$$
pointwise for some constant $C > 0$ and for sufficiently small $\eps,\delta > 0$. We use the notation
\begin{align}
g^{\eps,\delta} - f^{\eps,\delta} = O(\eps + \delta). \label{eq:bigO}
\end{align}
\end{obs}

\section{Derivatives on Future Contracts}\label{sec:option}

\subsection{First-Order Approximation for Future Prices}\label{sec:remarks_F_tT}

Here, we present the first-order approximation of future prices on mean-reverting assets. For a fixed maturity $T> 0$, we define
$$h^{\eps,\delta}(t,u,y,z,T) = \bE_{\bQ}[V_T \ | \ U_t = u, Y^{\eps}_t = y, Z^{\delta}_t = z],$$
and note that $F_{t,T} = h^{\eps,\delta}(t,U_t,Y_t^{\eps},Z_t^{\delta},T)$. We consider the formal expansion in powers of $\sqrt{\eps}$ and $\sqrt{\delta}$ of $h^{\eps,\delta}$:
$$h^{\eps,\delta}(t,u,y,z,T) = \sum_{i,j\geq 0} (\sqrt{\eps})^i (\sqrt{\delta})^j h_{i,j}(t,x,y,z,T).$$

We are interested in the first-order approximation of derivatives on mean-reverting assets, which is presented in \cite{jaimungal_futures} and \cite{multiscale_aymptotic_mean_rever}. Remember $Y^1$ denotes the process $Y^{\eps}$ with $\eps = 1$.

Applying the first-order approximation of future prices described in the aforesaid references, we choose the first terms of the above formal series to be
\begin{align}
\hspace{15pt} &h_0(t,u,z,T) = \exp\left\{ s(T) + m + (u - m)e^{-\kappa(T-t)} + \frac{\bar{\eta}^2(z)}{4\kappa}\left(1  - e^{-2\kappa(T - t)}\right)\right\}, \label{eq:h_0} \\
\hspace{15pt} &h_{1,0}(t,u,z,T) = g(t,T)V_{3}(z) \frac{\partial^3h_0}{\partial u^3}(t,u,z,T), \label{eq:h_10} \\
\hspace{15pt} &h_{0,1}(t,u,z,T) = f(t,T) V_{1}(z) \frac{\partial^3 h_0}{\partial u^3}(t,u,z,T), \label{eq:h_01}
\end{align}
where, denoting the averaging with respect to the invariant distribution of $Y^1$ by $\langle\cdot\rangle$, we have
\begin{align}
\bar{\eta}^2(z) &= \langle \eta^2(\cdot,z)\rangle, \label{eq:eta_bar} \\
V_3(z) &= - \frac{\rhor_{1}}{2} \left\langle \eta(\cdot, z) \beta(\cdot) \frac{\partial \phi}{\partial y}(\cdot, z) \right\rangle, \label{eq:V_3_h} \\
V_1(z) &= \rhor_{2} g(z) \langle \eta(\cdot,z) \rangle \bar{\eta}(z) \bar{\eta}'(z), \label{eq:V_1_h} \\
f(t,T) &= \frac{e^{3\kappa (T-t)} - e^{2\kappa (T-t)} }{2\kappa^2} - \frac{e^{3\kappa (T-t)} - 1}{6\kappa^2}, \label{eq:f_h} \\
g(t,T) &= \frac{e^{-3\kappa(T-t)} -1}{3\kappa} \label{eq:g_h},
\end{align}
and $\phi(y,z)$ is the solution of the Poisson equation
\begin{align}\label{eq:poisson_eq}
\cL_0 \phi(y,z) = \eta^2(y,z) - \bar{\eta}^2(z),
\end{align}
with $\cL_0$ being the infinitesimal generator of $Y^1$. Moreover, we may assume $h_{1,1}$ does not depend on $y$ and choose
\begin{align}\label{eq:h_20}
h_{2,0}(t,u,y,z,T) = -\frac{1}{2} \phi(y,z) \frac{\partial^2 h_0}{\partial u^2}(t,u,z,T) + c(t,u,z,T),
\end{align}
for some function $c$ that does not depend on $y$. Under all these choices and some regularity conditions similar to the ones presented in Theorem \ref{thm:accuracy} at the end of this section, as it was shown in \cite{multiscale_aymptotic_mean_rever} and \cite{jaimungal_futures}, we have
\begin{align*}
h^{\eps,\delta}(t,u,y,z) =&\,\, h_0(t,u,z,T) + \sqrt{\eps} h_{1,0}(t,u,z,T) + \sqrt{\delta} h_{0,1}(t,u,z,T) + O(\eps + \delta).
\end{align*}
Furthermore, the following simplifications hold:
$$h_{1,0}(t,u,z,T) = g(t,T)V_{3}(z) e^{-3\kappa(T-t)} h_0(t,u,z,T),$$
and
$$h_{0,1}(t,u,z,T) = f(t,T) V_{1}(z) e^{-3\kappa(T-t)} h_0(t,u,z,T).$$

\subsection{The Dynamics of the Future Prices}

In this section, we will derive the SDE describing the dynamics of $F_{t,T}$ and write its coefficients as functions of $F_{t,T}$. Since $F_{t,T}$ is a martingale under $\bQ$, its dynamics has no drift and hence, applying It\^o's Formula to $F_{t,T} = h^{\eps,\delta}(t,U_t,Y^{\eps}_t,Z^{\delta}_t,T)$, we get
\begin{align*}
dF_{t,T} &= \frac{\partial h^{\eps,\delta}}{\partial u}(t,U_t,Y_t^{\eps},Z_t^{\delta},T) \eta(Y^{\eps}_t,Z^{\delta}_t) dW_t^{(0)} \\
&+ \frac{1}{\sqrt{\eps}}\frac{\partial h^{\eps,\delta}}{\partial y}(t,U_t,Y_t^{\eps},Z_t^{\delta},T) \beta(Y_t^{\eps}) dW_t^{(1)} \\
&+\sqrt{\delta}\frac{\partial h^{\eps,\delta}}{\partial z}(t,U_t,Y_t^{\eps},Z_t^{\delta},T) g(Z^{\delta}_t) dW_t^{(2)}.
\end{align*}

We are interested in derivatives contracts on $F_{t,T}$ and in applying the perturbation method to approximate their prices. Thus, we will rewrite the SDE above with all coefficients depending on $F_{t,T}$ instead of $U_t$. In order to proceed, we assume we can invert $h^{\eps,\delta}$ with respect to $u$ for fixed $\eps, \delta, y, z$ and $T$, i.e. there exists a function $H^{\eps,\delta}(t,x,y,z,T)$ such that
$$H^{\eps,\delta}(t,\cdot,y,z,T) = (h^{\eps,\delta}(t,\cdot,y,z,T))^{-1}  \mbox{ .}$$
Since $h_0(t,u,z)$ given by \eqref{eq:h_0} is invertible in $u$, at least for small $\eps$ and $\delta$, this inversion is not a strong assumption on our model. The asymptotic analysis of $H^{\eps,\delta}$ is given in the following lemma.

\begin{lema}\label{lemma:inv}

If we choose $H_0$, $H_{1,0}$, $H_{0,1}$ to be
\begin{itemize}

\item[(i)] $H_0(t,\cdot,z,T) = (h_0(t,\cdot,z,T))^{-1},$\\

\item[(ii)] $\ds H_{1,0}(t,x,z,T) = -\frac{h_{1,0}(t,H_0(t,x,z,T),z,T)}{\ds\frac{\partial h_0}{\partial u}(t,H_0(t,x,z,T),z,T)},$\\

\item[(iii)] $\ds H_{0,1}(t,x,z,T) = -\frac{h_{0,1}(t,H_0(t,x,z,T),z,T)}{\ds\frac{\partial h_0}{\partial u}(t,H_0(t,x,z,T),z,T)},$

\end{itemize}
where $h_{1,0}$ and $h_{0,1}$ are given by \eqref{eq:h_10} and \eqref{eq:h_01} respectively, then, we have
\begin{align*}
H^{\eps,\delta}(t,x,y,z,T) &= H_0(t,x,z,T) + \sqrt{\eps} H_{1,0}(t,x,z,T) + \sqrt{\delta} H_{0,1}(t,x,z,T) + O(\eps+\delta).
\end{align*}

\end{lema}

\begin{proof}
The derivation is straightforward.
\end{proof}

Notice that $H^{\eps,\delta}(t,F_{t,T},y,z,T) = U_t$ and if we define
\begin{align}
\psi^{\eps,\delta}_1(t,x,y,z,T) &= \frac{\partial h^{\eps,\delta}}{\partial u}(t,H^{\eps,\delta}(t,x,y,z,T),y,z,T), \label{eq:psi_1}\\
\psi^{\eps,\delta}_2(t,x,y,z,T) &= \frac{\partial h^{\eps,\delta}}{\partial y}(t,H^{\eps,\delta}(t,x,y,z,T),y,z,T), \label{eq:psi_2}\\
\psi^{\eps,\delta}_3(t,x,y,z,T) &= \frac{\partial h^{\eps,\delta}}{\partial z}(t,H^{\eps,\delta}(t,x,y,z,T),y,z,T), \label{eq:psi_3}
\end{align}
we obtain the desired SDE for $F_{t,T}$
\begin{align}
dF_{t,T} &= \psi^{\eps,\delta}_1(t,F_{t,T},Y_t^{\eps},Z_t^{\delta},T) \eta(Y^{\eps}_t,Z^{\delta}_t) dW_t^{(0)}
\label{eq:dF}\\
&+ \frac{1}{\sqrt{\eps}}\psi^{\eps,\delta}_2(t,F_{t,T},Y_t^{\eps},Z_t^{\delta},T) \beta(Y_t^{\eps}) dW_t^{(1)} \nonumber\\
&+\sqrt{\delta}\psi^{\eps,\delta}_3(t,F_{t,T},Y_t^{\eps},Z_t^{\delta},T) g(Z^{\delta}_t) dW_t^{(2)}.\nonumber
\end{align}

\subsection{A Pricing PDE for Derivatives on Future Contracts}

We now fix a future contract on $V$ with maturity $T$ and consider a European derivative with maturity $T_0 < T$ and whose payoff $\varphi$ depends only on the terminal value $F_{T_0,T}$. A no-arbitrage price for this derivative on $F_{t,T}$ is given by
$$P^{\eps,\delta}(t,x,y,z,T) = \bE_{\bQ}[e^{-r(T_0-t)}\varphi(F_{T_0,T}) \ | \ F_{t,T} = x, Y^{\eps}_t = y, Z^{\delta}_t = z],$$
where $\bQ$ is the risk-neutral probability discussed in Section \ref{sec:model} and we are using the fact that $(F_{t,T},Y^{\eps}_t,Z^{\delta}_t)$ is a Markov process. In this section we derive a PDE for $P^{\eps,\delta}$. Recall that $F_{t,T}$ follows Equation \eqref{eq:dF}, where $Y_t^{\eps}$ and $Z_t^{\delta}$ are given in \eqref{eq:sde_risk_neutral}.
Then, we write the infinitesimal generator $\cL^{\eps,\delta}$ of $(F_{t,T},Y^{\eps}_t,Z^{\delta}_t)$, where, for simplicity of notation, we will drop the variables $(t,x,y,z,T)$ of $\psi_i^{\eps,\delta}$, $i=1,2,3$,
\begin{align}\label{eq:cl_eps_delta}
\hskip .5cm \cL^{\eps,\delta} &= \frac{1}{\eps} \left( \cL_0 + \frac{1}{2}(\psi^{\eps,\delta}_2)^2 \beta^2(y) \frac{\partial^2}{\partial x^2} + \psi^{\eps,\delta}_2 \beta^2(y) \frac{\partial^2}{\partial x \partial y} \right) \\ \nonumber
&+ \frac{1}{\sqrt{\eps}}\left(\rhor_{1} \psi^{\eps,\delta}_1 \psi^{\eps,\delta}_2 \eta(y,z)\beta(y) \frac{\partial^2}{\partial x^2} + \rhor_{1} \psi^{\eps,\delta}_1 \eta(y,z)\beta(y) \frac{\partial^2}{\partial x \partial y}\right) \\ \nonumber
&+ \frac{\partial}{\partial t} + \frac{1}{2}(\psi^{\eps,\delta}_1)^2 \eta^2(y,z) \frac{\partial^2}{\partial x^2} - r \cdot \\ \nonumber
&+ \sqrt{\delta}\left( \rhor_{2} \psi^{\eps,\delta}_1 \psi^{\eps,\delta}_3 \eta(y,z)g(z) \frac{\partial^2}{\partial x^2} + \rhor_{2} \psi^{\eps,\delta}_1 \eta(y,z)g(z) \frac{\partial^2}{\partial x \partial z}\right) \\ \nonumber
&+ \delta \left( \cM_2 + \frac{1}{2} (\psi^{\eps,\delta}_3)^2 g^2(z) \frac{\partial^2}{\partial x^2} + \psi^{\eps,\delta}_3 g^2(z) \frac{\partial^2}{\partial x \partial z} \right) \\ \nonumber
&+ \sqrt{\frac{\delta}{\eps}} \left(\rhor_{12} \psi^{\eps,\delta}_2 \psi^{\eps,\delta}_3 \beta(y)g(z) \frac{\partial^2}{\partial x^2}+ \rhor_{12} \psi^{\eps,\delta}_3 \beta(y)g(z) \frac{\partial^2}{\partial x \partial y} \right.\\
& \hskip 1.2cm \left.+ \rhor_{12} \psi^{\eps,\delta}_2 \beta(y)g(z) \frac{\partial^2}{\partial x \partial z} + \rhor_{12} \beta(y) g(z) \frac{\partial^2}{\partial y \partial z}\right),\nonumber
\end{align}
where
\begin{align}
\cL_0 = \frac{1}{2} \beta^2(y) \frac{\partial^2}{\partial y^2} + \alpha(y) \frac{\partial}{\partial y}, \label{eq:cL_0} \\
\cM_2 = \frac{1}{2} g^2(z) \frac{\partial^2}{\partial z^2} + c(z) \frac{\partial}{\partial z}. \label{eq:cM_2}
\end{align}

It is well-known that under some mild conditions, by Feynman-Kac's Formula, $P^{\eps,\delta}$ satisfies the pricing PDE
\begin{align}\label{eq:pricing_PDE}
\left\{
\begin{array}{l}
 \cL^{\eps,\delta}P^{\eps,\delta}(t,x,y,z,T) = 0, \\ \\
 P^{\eps,\delta}(T_0,x,y,z,T) = \varphi(x).
\end{array}
\right.
\end{align}

\subsection{Perturbation Framework}\label{sec:perturbation}

We will now develop the formal singular and regular perturbation analysis for European derivatives on $F_{t,T}$ following the method outlined in \cite{fouque_stochastic_vol_new}. However, in our case we have a fundamental difference: the coefficients of the differential operator $\cL^{\eps,\delta}$, given by Equation \eqref{eq:cl_eps_delta}, depend on $\eps$ and $\delta$ in an intricate way. In particular, the term corresponding to the factor $\eps^{-1}$ is not simply of order $\eps^{-1}$. To circumvent this problem, we will expand the coefficients in powers of $\eps$ and $\delta$ and then collect the correct terms for each order. Therefore, it will be necessary to compute some terms of the expansion of $\psi_i^{\eps,\delta}$. All the details for this expansion are given in the Appendix \ref{sec:appendix_pde_exp} and the final result is:
$$\cL^{\eps,\delta} = \frac{1}{\eps} \cL_0 + \frac{1}{\sqrt{\eps}} \cL_1 + \cL_2 + \sqrt{\eps} \cL_3 + \sqrt{\delta} \cM_1 + \sqrt{\frac{\delta}{\eps}} \cM_3 + \cdots,$$
where $\cL_0$ is given by \eqref{eq:cL_0} and
\begin{align}
\hskip 1cm \cL_1 &= \rhor_{1} e^{-\kappa(T-t)} \eta(y,z)\beta(y) x \frac{\partial^2}{\partial x \partial y}, \label{eq:cL_1} \\
\hskip 1cm \cL_2 &= \frac{\partial}{\partial t} + \frac{1}{2} e^{-2\kappa(T-t)} \eta^2(y,z) x^2 \frac{\partial^2}{\partial x^2} -r \cdot \label{eq:cL_2}\\
&- \frac{1}{2}e^{-2\kappa(T-t)} \frac{\partial \phi}{\partial y}(y,z) \beta^2(y) x \frac{\partial^2}{\partial x \partial y}, \nonumber \\
\hskip 1cm \cM_3 &= \rhor_{12} \frac{(1 - e^{-2\kappa (T-t)}) }{2\kappa} \beta(y)g(z) \bar{\eta}(z) \bar{\eta}'(z) x \frac{\partial^2}{\partial x \partial y} \label{eq:cM_3} \\
&+ \rhor_{12} \beta(y) g(z) \frac{\partial^2}{\partial y \partial z}, \nonumber \\
\hskip 1cm \cL_3 &= (\psi_{2,3,0}(t,x,y,z,T) \beta^2(y) \label{eq:cL_3}\\
&+ \rhor_{1} \psi_{1,2,0}(t,x,y,z,T) \eta(y,z) \beta(y)) \frac{\partial^2}{\partial x \partial y} \nonumber \\
&- \rhor_{1}\frac{1}{2}e^{-3\kappa(T-t)} \frac{\partial \phi}{\partial y}(y,z) \eta(y,z) \beta(y) x^2 \frac{\partial^2}{\partial x^2},\nonumber \\
\hskip 1cm \cM_1 &= \rhor_{2} e^{-\kappa(T-t)} \frac{(1 - e^{-2\kappa (T-t)})}{2\kappa} \eta(y,z)g(z) \bar{\eta}(z) \bar{\eta}'(z) x^2 \frac{\partial^2}{\partial x^2} \label{eq:cM_1} \\
& + \rhor_{2} e^{-\kappa(T-t)} \eta(y,z)g(z) x\frac{\partial^2}{\partial x \partial z} + (\psi_{2,2,1}(t,x,y,z,T) \beta^2(y) \nonumber \\
& + \rhor_{2} \psi_{1,1,1}(t,x,T) \eta(y,z)\beta(y)) \frac{\partial^2}{\partial x \partial y} .\nonumber
\end{align}
The fundamental difference with the situation described in \cite{fouque_stochastic_vol_new} then materializes in one term: \textit{the differential operator $\cL_3$ which contributes to the order $\sqrt{\eps}$ in the expansion of $\cL^{\eps,\delta}$}. Also, observe that the coefficients of these operators are time dependent which complicates the asymptotic analysis. This difficulty has also been dealt with in \cite{matcycles}.

\subsection{Formal Derivation of the First-Order Approximation}

Let us formally write $P^{\eps,\delta}$ in powers of $\sqrt{\delta}$ and $\sqrt{\eps}$,
$$P^{\eps,\delta} = \sum_{m,k \geq 0} (\sqrt{\eps})^k (\sqrt{\delta})^m P_{k,m},$$
and denote $P_{0,0}$ simply by $P_0$ where we assume that, at maturity $T_0$, $P_0(T_0,x,y,z,T) =\varphi(x)$.
 We are interested in determining $P_0$, $P_{1,0}$ and $P_{0,1}$. We follow the method presented in \cite{fouque_stochastic_vol_new} with some minor modifications in order to take into account the new term $\cL_3$.

In order to compute the leading term $P_0$ and $P_{1,0}$, we set to be zero the following terms of the expansion of $\cL^{\eps,\delta}P^{\eps,\delta}$:
\begin{align}
(-1,0): & \ \cL_0P_0 = 0, \label{eq:edp_eps_1}\\
(-1/2,0):& \ \cL_0 P_{1, 0}+ \cL_1P_0 = 0 \label{eq:edp_eps_2},\\
(0,0):& \ \cL_0 P_{2,0} + \cL_1 P_{1,0} + \cL_2P_0 = 0 \label{eq:edp_eps_3},\\
 (1/2,0):& \ \cL_0 P_{3,0} + \cL_1 P_{2,0} + \cL_2P_{1,0} + \cL_3 P_0= 0, \label{eq:edp_eps_4}
\end{align}
where we are using the notation $(i,j)$ to denote the term of $i$th order in $\eps$ and $j$th in $\delta$.

\subsubsection{Computing $P_0$}

We seek a function $P_0 = P_0(t,x,z,T)$, independent of $y$, so that the Equation \eqref{eq:edp_eps_1} is satisfied. Since $\cL_1$ takes derivative with respect to $y$, $\cL_1P_0 = 0$. Thus the second \eqref{eq:edp_eps_2} becomes $\cL_0 P_{1,0} = 0$ and for the same reason as before, we seek a function $P_{1,0} = P_{1,0}(t,x,z,T)$ independent of $y$. The $(0,0)$-order equation \eqref{eq:edp_eps_3} becomes
$$\cL_0P_{2,0} + \cancelto{\scriptstyle 0}{\cL_1 P_{1,0}} + \cL_2P_0 = 0,$$
which is a Poisson equation for $P_{2,0}$ with solvability condition
$$\langle \cL_2P_0 \rangle = 0,$$
where $\langle \cdot \rangle$ is the average under the invariant measure of $\cL_0$. For more details on Poisson equations, see \cite[Section 3.2]{fouque_stochastic_vol_new}. Define now
\begin{align}\label{eq:cL_B}
\cL_B(\sigma) = \frac{\partial}{\partial t} + \frac{1}{2} \sigma^2 x^2\frac{\partial^2}{\partial x^2} -r \cdot
\end{align}
and
\begin{align}\label{eq:sigma_t}
\sigma(t,y,z,T) = e^{-\kappa(T-t)} \eta(y,z),
\end{align}
where we are using the notation $\cL_B(\sigma)$ for the Black differential operator with volatility $\sigma$. Since $P_0$ does not depend on $y$ and by the form of $\cL_2$ given in \eqref{eq:cL_2}, the solvability condition becomes
$$\langle \cL_2P_0 \rangle=\langle \cL_B(\sigma(t,y,z,T)) \rangle P_0 = 0.$$
Note that
$$\langle \cL_B(\sigma(t,y,z,T)) \rangle = \frac{\partial}{\partial t} + \frac{1}{2} x^2 \left\langle \sigma^2(t,\cdot,z,T) \right\rangle \frac{\partial^2}{\partial x^2} - r\cdot = \cL_B(\bar{\sigma}(t,z,T)),$$
where
\begin{align}\label{eq:bar_sigma}
\bar{\sigma}^2(t,z,T) = \left\langle \sigma^2(t,\cdot,z,T) \right\rangle = e^{-2\kappa(T-t)} \bar{\eta}^2(z),
\end{align}
with $\bar{\eta}(z)$ defined in \eqref{eq:eta_bar}. Therefore, we choose $P_0$ to satisfy the PDE
\begin{align} \label{eq:edp_p0}
\left\{
\begin{array}{l}
 \cL_B(\bar{\sigma}(t,z,T))P_0(t,x,z,T) = 0, \\ \\
 P_0(T_0,x,z,T) = \varphi(x).
\end{array}
\right.
\end{align}

Note also that $\cL_B(\bar{\sigma}(t,z,T))$ is the Black differential operator with time-varying volatility $\bar{\sigma}(t,z,T)$ and hence, if we define the time-averaged volatility, $\bar{\sigma}_{t,T_0}(z,T)$, by the formula
\begin{align}
\bar{\sigma}^2_{t,T_0}(z,T) &= \frac{1}{T_0-t}\int_t^{T_0} \bar{\sigma}^2(u,z,T)du \label{eq:sigma_bar_bar}\\
&= \bar{\eta}^2(z)\left( \frac{e^{-2\kappa (T-T_0)} - e^{-2\kappa (T-t)}}{2\kappa(T_0-t)}\right), \nonumber
\end{align}
we can write
$$P_0(t,x,z,T) = P_B(t,x, \bar{\sigma}_{t,T_0}(z,T)),$$
where $P_B(t,x,\sigma)$ is the price at $(t,x)$ of the European derivative with maturity $T_0$ and payoff function $\varphi$ in the Black model with constant volatility $\sigma$.

In order to simplify notation here and in what follows, we define
\begin{align}
\lambda(t,T_0,T,\kappa) &= \frac{e^{-\kappa (T-T_0)} - e^{-\kappa (T-t)}}{\kappa(T_0-t)}. \label{eq:lambda}
\end{align}
Therefore,
\begin{align}\label{eq:sigma_bar_bar_lambda}
\bar{\sigma}^2_{t,T_0}(z,T) = \bar{\eta}^2(z)\lambda_{\sigma}^2(t,T_0,T,\kappa),
\end{align}
where
\begin{align}
\lambda_{\sigma}(t,T_0,T,\kappa) &= \sqrt{\lambda(t,T_0,T,2\kappa)}. \label{eq:lambda_sigma}
\end{align}

\subsubsection{Computing $P^{\eps}_{1,0}$}

By the $(0,0)$-order equation \eqref{eq:edp_eps_3}, we get the formula
\begin{align}\label{eq:p20}
P_{2,0} = -\cL_0^{-1}(\cL_B(\sigma) - \cL_B(\bar{\sigma}))P_0 + c(t,x,z,T),
\end{align}
for some function $c$ which does not depend on $y$. Denote by $\phi(y,z)$ a solution of the Poisson equation
$$\cL_0 \phi(y,z) = \eta^2(y,z) - \bar{\eta}^2(z).$$
Hence,
\begin{align*}
\cL_0^{-1}(\cL_B(\sigma) - \cL_B(\bar{\sigma})) &= \cL_0^{-1} \left(\frac{1}{2} (\sigma^2(t,y,z,T) - \bar{\sigma}^2(t,z,T)) x^2 \frac{\partial^2}{\partial x^2} \right) \\
&= \frac{1}{2}e^{-2\kappa(T-t)} \cL_0^{-1} (\eta^2(y,z) - \bar{\eta}^2(z)) x^2 \frac{\partial^2}{\partial x^2} \\
&= \frac{1}{2} e^{-2\kappa(T-t)} \phi(y,z) D_2,
\end{align*}
where we use the notation
\begin{align}\label{eq:D_k}
D_k = x^k \frac{\partial^k}{\partial x^k}.
\end{align}
From the $(1/2,0)$-order equation \eqref{eq:edp_eps_4}, which is a Poisson equation for $P_{3,0}$, we get the solvability condition
\begin{align}
\langle \cL_1 P_{2,0} + \cL_2 P_{1,0} + \cL_3 P_0 \rangle = 0.\label{eq:solvability}
\end{align}
Using formula \eqref{eq:p20} for $P_{2,0}$ and formula \eqref{eq:cL_1} for $\cL_1$, we get
\begin{align*}
\cL_1 P_{2,0} &= - \cL_1 \cL_0^{-1}(\cL_B(\sigma) - \cL_B(\bar{\sigma}))P_0 \\
&= -\rhor_{1} e^{-\kappa(T-t)} \eta(y,z)\beta(y) x\frac{\partial^2}{\partial x \partial y} \left(\frac{1}{2}e^{-2\kappa(T-t)} \phi(y,z) D_2 P_0\right) \\
&= -\rhor_{1} e^{-\kappa(T-t)} x\eta(y,z)\beta(y) \frac{\partial}{\partial x} \frac{\partial}{\partial y} \left(\frac{1}{2}e^{-2\kappa(T-t)} \phi(y,z) D_2 P_0\right) \\
&= -\frac{1}{2}\rhor_{1} e^{-3\kappa(T-t)} \eta(y,z)\beta(y) \frac{\partial \phi}{\partial y}(y,z) D_1 D_2 P_0.
\end{align*}
We also have by equation \eqref{eq:cL_2}
$$\cL_2 P_{1,0} = \frac{\partial P_{1,0}}{\partial t} + \frac{1}{2} \sigma^2(t,y,z,T) x^2\frac{\partial^2 P_{1,0}}{\partial x^2} - r P_{1,0},$$
and from \eqref{eq:cL_3}
\begin{align*}
\cL_3 P_{0} &= -\frac{1}{2}\rhor_1 e^{-3\kappa(T-t)} \frac{\partial \phi}{\partial y}(y,z) \eta(y,z) \beta(y) D_2 P_0.
\end{align*}
Combining these equations, we get
\begin{align*}
\cL_1 P_{2,0} + \cL_2 P_{1,0} + \cL_3 P_0 &= v_3(t,y,z,T) D_2P_0 \\
&+ v_3(t,y,z,T) D_1D_2 P_0 + \cL_B(\sigma(t,y,z,T)) P_{1,0},
\end{align*}
where
$$v_3(t,y,z,T) = -\frac{1}{2}\rhor_{1} e^{-3\kappa(T-t)} \frac{\partial \phi}{\partial y}(y,z) \eta(y,z) \beta(y).$$

Therefore, averaging with respect to the invariant distribution of $Y^1$, we deduce from \eqref{eq:solvability} that $P_{1,0}^{\eps} = \sqrt{\eps}P_{1,0}$ satisfies the PDE:
\begin{align}\label{eq:pde_p_1_eps}
\left\{
\begin{array}{l}
 \cL_B(\bar{\sigma}(t,z,T))P^{\eps}_{1,0}(t,x,z,T) = -f(t,T)\cA^{\eps}P_0(t,x,z,T), \\ \\
 P^{\eps}_{1,0}(T_0,x,z,T) = 0,
\end{array}
\right.
\end{align}
where
\begin{align}
\cA^{\eps} &= V_3^{\eps}(z)(D_1D_2 + D_2), \label{eq:Aeps} \\
f(t,T) &= e^{-3\kappa(T-t)}, \label{eq:Aeps_f} \\
V_3^{\eps}(z) &= -\sqrt{\eps}\frac{1}{2}\rhor_{1} \left\langle \frac{\partial \phi}{\partial y}(\cdot,z) \eta(\cdot,z) \beta \right\rangle. \label{eq:V2eps}
\end{align}
The linear PDE \eqref{eq:pde_p_1_eps} is solved explicitly:
\begin{align}
\hskip .8cm P_{1,0}^{\eps}(t,x,z,T) =  (T_0-t)\lambda_3(t,T_0,T,\kappa)V_3^{\eps}(z)(D_1D_2 + D_2)P_B(t,x,\bar{\sigma}_{t,T_0}(z,T)), \label{eq:P10_eps}
\end{align}
where
\begin{align} \label{eq:lambda_3}
\lambda_3(t,T_0,T,\kappa) = \lambda(t,T_0,T,3\kappa),
\end{align}
and $\lambda$ is defined by \eqref{eq:lambda}. To see this,
note that the operator $\cA^{\eps}$ given by \eqref{eq:Aeps}, and the operator $\cL_B(\bar{\sigma}(t,z,T))$ given by \eqref{eq:cL_B} and \eqref{eq:bar_sigma}, commute and therefore, the solution of the PDE \eqref{eq:pde_p_1_eps} is given by
$$P_{1,0}^{\eps}(t,x,z,T) = \left(\int_t^{T_0} f(u,T)du \right) \cA^{\eps}P_0(t,x,z,T).$$
Thus, solving the above integral, we get \eqref{eq:P10_eps}.


\subsubsection{Computing $P^{\delta}_{0,1}$}\label{sec:comp_P01_delta}

In order to compute $P_{0,1}$, we need to consider terms with order 1/2 in $\delta$, more explicitly the following ones:
\begin{align}
\hskip 1cm (-1,1/2):& \ \cL_0 P_{0,1} = 0,\label{eq:edp_delta_1}\\
\hskip 1cm (-1/2,1/2):& \ \cL_0P_{1,1}+ \cL_1P_{0,1} + \cM_3P_0 = 0, \label{eq:edp_delta_2}\\
\hskip 1cm (0,1/2):& \ \cL_0P_{2,1} + \cL_1 P_{1,1} + \cM_3P_{1,0} + \cL_2P_{0,1} + \cM_1P_0 = 0. \label{eq:edp_delta_3}
\end{align}

Recall that $\cL_1$ and $\cM_3$ as defined by Equations \eqref{eq:cL_1} and \eqref{eq:cM_3} take derivative with respect to $y$. Choosing $P_{0,1} = P_{0,1}(t,x,z,T)$ and $P_{1,1} = P_{1,1}(t,x,z,T)$ independent of $y$, the first two equations \eqref{eq:edp_delta_1} and \eqref{eq:edp_delta_2} are satisfied.
The last equation \eqref{eq:edp_delta_3} becomes
$$\cL_0P_{2,1} + \cL_2P_{0,1} + \cM_1 P_0 = 0,$$
and thus the solvability condition for this Poisson equation for $P_{2,1}$ is
$$\langle \cL_2P_{0,1} + \cM_1 P_0 \rangle = 0.$$
From \eqref{eq:cM_1} we have
\begin{align*}
\cM_1P_0 &= \rhor_{2} \bar{\eta}(z)\bar{\eta}'(z)\frac{(1 - e^{-2\kappa (T-t)})}{2\kappa}e^{-\kappa(T-t)} \eta(y,z)g(z) D_2P_0 \\
&+ \rhor_{2} e^{-\kappa (T-t)} \eta(y,z)g(z) D_1\frac{\partial}{\partial z}P_0,
\end{align*}
and then, if we write $P^{\delta}_{0,1}(t,x,z,T) = \sqrt{\delta}P_{0,1}(t,x,z,T)$, the solvability condition above can be written as
\begin{align}\label{eq:pde_p_1_delta}
\left\{
\begin{array}{l}
 \cL_B(\bar{\sigma}(t,z,T)) P_{0,1}^{\delta} = -f_0(t,T)\cA_0^{\delta}P_0 -f_1(t,T) \cA_1^{\delta} P_0, \\ \\
 P_{0,1}^{\delta}(T_0,x,z,T) = 0,
\end{array}
\right.
\end{align}
where
\begin{align}
\cA_0^{\delta} &= V_0^{\delta}(z) D_2, \label{eq:A0delta} \\
\cA_1^{\delta} &= V_1^{\delta}(z) D_1\frac{\partial}{\partial z}, \label{eq:A1delta} \\
V_0^{\delta}(z) &= \sqrt{\delta} \frac{1}{2\kappa}\rhor_{2} \langle \eta(\cdot,z) \rangle g(z) \bar{\eta}(z)\bar{\eta}'(z), \label{eq:V0delta} \\
f_0(t,T) &= e^{-\kappa(T-t)} - e^{-3\kappa (T-t)}, \label{eq:f_0} \\
V_1^{\delta}(z) &= \sqrt{\delta} \rhor_{2} \langle \eta(\cdot,z) \rangle g(z), \label{eq:V1delta} \\
f_1(t,T) &= e^{-\kappa (T-t)}. \label{eq:f_1}
\end{align}
The solution for this PDE can be also explicitly computed \vspace{-15pt}
\begin{align}
P_{0,1}^{\delta}(t,x,z,T) = (T_0-t)V_0^{\delta}(z)(\lambda_0(t,T_0,T,\kappa)D_2  +\lambda_1(t,T_0,T,\kappa)D_1 D_2)P_B(t,x,\bar{\sigma}_{t,T_0}(z,T)), \label{eq:P01_delta}
\end{align}
where
\begin{align}
\lambda_0(t,T_0,T,\kappa) &= \lambda(t,T_0,T,\kappa) -\lambda(t,T_0,T,3\kappa), \label{eq:lambda_0}\\
\lambda_1(t,T_0,T,\kappa) &= e^{-2\kappa(T-T_0)}\lambda(t,T_0,T,\kappa) -\lambda(t,T_0,T,3\kappa). \label{eq:lambda_1}
\end{align}

The details of this computation are given in the Appendix \ref{sec:p_1_delta_comp}.

\subsection{Summary and Some Remarks}\label{sec:summary}

We now summarize the formulas involved in the first-order asymptotic expansion of the price of European derivative on futures. We recall that, as before, $D_k = x^k \partial / \partial x_k$. We have formally derived the first-order approximation of $P^{\eps,\delta}$:
$$P^{\eps,\delta}\approx P_0+P_{1,0}^{\eps}+P_{0,1}^{\delta},$$
with
\begin{align}
&P_0(t,x,z,T) = P_B(t,x, \bar{\sigma}_{t,T_0}(z,T)), \label{eq:P0_summ} \\
&P_{1,0}^{\eps}(t,x,z,T) = (T_0-t)\lambda_3(t,T_0,T,\kappa)V_3^{\eps}(z)(D_2 + D_1D_2)P_B(t,x,\bar{\sigma}_{t,T_0}(z,T)), \label{eq:P10_eps_summ}  \\
&P_{0,1}^{\delta}(t,x,z,T) = (T_0-t)V_0^{\delta}(z)(\lambda_0(t,T_0,T,\kappa)D_2  +\lambda_1(t,T_0,T,\kappa)D_1 D_2)P_B(t,x,\bar{\sigma}_{t,T_0}(z,T)), \label{eq:P01_delta}
\end{align}
where
\begin{align}
\bar{\eta}^2(z) &= \langle \eta^2(\cdot,z) \rangle, \label{eq:eta_bar_summ} \\
V_3^{\eps}(z) &= -\sqrt{\eps}\frac{1}{2}\rhor_{1} \left\langle \frac{\partial \phi}{\partial y}(\cdot,z) \eta(\cdot,z) \beta \right\rangle, \label{eq:V2eps_summ} \\
V_0^{\delta}(z) &= \sqrt{\delta} \frac{1}{2\kappa}\rhor_{2} \langle \eta(\cdot,z) \rangle g(z) \bar{\eta}(z)\bar{\eta}'(z), \label{eq:V0delta_summ} \\
\lambda(t,T_0,T,\kappa) &= \frac{e^{-\kappa(T - T_0)} - e^{-\kappa(T-t)}}{\kappa(T_0-t)}, \label{eq:lambda_summ} \\
\lambda_3(t,T_0,T,\kappa) &= \lambda(t,T_0,T,3\kappa), \label{eq:lambda3_summ} \\
\lambda_0(t,T_0,T,\kappa) &= \lambda(t,T_0,T,\kappa) -\lambda(t,T_0,T,3\kappa), \label{eq:lambda_0_summ}\\
\lambda_1(t,T_0,T,\kappa) &= e^{-2\kappa(T-T_0)}\lambda(t,T_0,T,\kappa) -\lambda(t,T_0,T,3\kappa), \label{eq:lambda_1_summ} \\
\lambda_{\sigma}^2(t,T_0,T,\kappa) & = \lambda(t,T_0,T,2\kappa). \label{eq:lambda_sigma_summ} \\
\bar{\sigma}^2_{t,T_0}(z,T) & = \bar{\eta}^2(z)\lambda_{\sigma}^2(t,T_0,T,\kappa). \label{eq:sigma_bar_bar_summ}
\end{align}

A valuable feature of the perturbation method is that in order to compute the first-order approximation, we only need the values of the \textit{group market parameters}
$$(\kappa, \bar{\eta}(z), V_0^{\delta}(z), V_3^{\eps}(z)).$$
This feature can also be seen as model independence and robustness of this approximation: under the regularity conditions stated in Theorem \ref{thm:accuracy}, this approximation is independent of the particular form of the coefficients describing the process $Y^{\eps}$ and $Z^{\delta}$, i.e. the functions $\alpha$, $\beta$, $c$ and $g$ involved in the model \eqref{eq:sde_risk_neutral}.

From now on we will use the following notation
\begin{align}
\bar{P}(t,x,z,T) &= P_0(t,x,z,T), \label{eq:Pbar_B} \\
\bar{P}^{\eps,\delta}(t,x,z,T) &= P_{1,0}^{\eps}(t,x,z,T) + P_{0,1}^{\delta}(t,x,z,T), \label{eq:Pbar_B_epsdelta} \\
\end{align}

\subsection{Accuracy of the Approximation}

We now state the precise accuracy result for the formal approximation determined in the previous sections.
All the reasoning in Section \ref{sec:perturbation} is only a formal procedure  and a well-thought \textit{choice} for the proposed first-order approximation. The next result establishes the order of accuracy of this approximation and justifies {\em a posteriori} the choices made earlier.

\begin{teo}\label{thm:accuracy}

We assume

\begin{itemize}

\item[(i)] Existence and uniqueness of the SDE \eqref{eq:sde_risk_neutral} for any fixed $(\eps,\delta)$.

\item[(ii)] The process $Y^1$ with infinitesimal generator $\cL_0$ has a unique invariant distribution and is mean-reverting as in \cite[Section 3.2]{fouque_stochastic_vol_new}.

\item[(iii)] The function $\eta(y,z)$ is smooth in $z$ and such the solution $\phi$ to the Poisson equation \eqref{eq:poisson_eq} is at most polynomially growing.

\item[(iv)] The payoff function $\varphi$ is continuous and piecewise smooth.

\end{itemize}
Then,
$$P^{\eps,\delta}(t,x,y,z,T) = \widetilde{P}^{\eps,\delta}(t,x,z,T) + O(\eps + \delta).$$
\end{teo}

\begin{proof}

The proof is provided in Appendix \ref{sec:proof}.

\end{proof}

Observe that in the heuristic derivation of the approximation given by Equations \eqref{eq:Pbar_B} and \eqref{eq:Pbar_B_epsdelta} we did not use any additional smoothness assumption as it is assumed in \cite{jaimungal_futures} for the derivation of their approximation. Hence our theorem covers the case of \textit{call options}. Being able to apply this approximation to \textit{call options} is essential to the next section.

%
%

\section{Calibration}\label{sec:calibration}

In this section we will outline a procedure to calibrate the \textit{group market parameters} $(\kappa, \bar{\eta}(z), V_0^{\delta}(z), V_3^{\eps}(z))$ to available prices of call options on $F_{t,T}$. As one may conclude from \cite[Chapters 6 and 7]{fouque_stochastic_vol_new}, or from an application of Functional It\^o Calculus (\cite{fito_dupire}) to the perturbation analysis, the values of the group market parameters are the only parameters needed to price path-dependent or American options to the same order of accuracy. Therefore, once the group market parameters are calibrated to vanilla options, the \textit{same} parameters are used to price exotic derivatives. This is one of the most important characteristics of the perturbation theory. Additionally, as we will conclude from what follows, the clean first-order approximation derived in this paper together with the fact we are considering the future price as the variable allow the derivation of the simple calibration procedure of the model to Black implied volatilities. This was not achieved in previous works on this subject, see Section \ref{sec:comparison}.

\subsection{Approximate Call Prices on Future Contracts and Implied Volatilities}\label{sec:approximate}

We assume without loss of generality $t=0$ and consider a European call option on $F_{0,T}$ with maturity $T_0 \leq T$ and strike $K$, i.e. the payoff function is given by $\varphi(x) = (x - K)^+$. As we are interested in the calibration of the market group parameters to call prices at the fixed time $t=0$, we will drop the variables $(t,x)$ in the formulas and write the variables $(T_0,K)$ instead. We will also drop the variable $z$ since it should be understood as just a parameter. The Black formula for a $(T_0,K)$-call option is defined by
\begin{align}\label{eq:black_formula}
C_B(T_0,K, \sigma) = e^{-rT_0}(F_{0,T}\Phi(d_1(\sigma)) - K\Phi(d_2(\sigma))),
\end{align}
where
$$d_{1,2}(\sigma) = \ds \frac{\log(F_{0,T}/K) \pm \frac{\sigma^2}{2}T_0}{\sigma\sqrt{T_0}}.$$
Let us also denote
$$\bar{d}_{1,2} = d_{1,2}(\bar{\sigma}_{0,T_0}),$$
where $\bar{\sigma}_{0,T_0}$ is the time-averaged volatility defined in \eqref{eq:sigma_bar_bar}, and notice that Equation \eqref{eq:Pbar_B} satisfies
\begin{align}\label{eq:Pbar_CB}
\bar{P}(0,F_{0,T},z,T) = C_B(T_0,K, \bar{\sigma}_{0,T_0}).
\end{align}
The following relations between Greeks of the Black price are well-known and they will be essential in what follows:
$$\frac{\partial C_B}{\partial \sigma}(T_0,K,\sigma) = T_0 \sigma D_2C_B(T_0,K,\sigma),$$
and
\begin{align}
D_1 \frac{\partial C_B}{\partial \sigma}(T_0,K,\sigma) &= \left(1 - \frac{d_1}{\sigma\sqrt{T_0}}\right) \frac{\partial C_B}{\partial \sigma}(T_0,K,\sigma) \label{eq:D1Vega} \\
&= \left(\frac{1}{2} + \frac{\log(K/x)}{\sigma^2T_0}\right) \frac{\partial C_B}{\partial \sigma}(T_0,K,\sigma), \nonumber
\end{align}
where the operator $D_k$ is defined in \eqref{eq:D_k}. Using the Greeks relations stated above and \eqref{eq:Pbar_CB}, we are able to rewrite \eqref{eq:Pbar_B_epsdelta} as
\begin{align*}
\bar{P}^{\eps,\delta} &= \left(\frac{\lambda_3(T_0,T,\kappa)}{\bar{\sigma}_{0,T_0}}V_3^{\eps}(z) + \frac{\lambda_3(T_0,T,\kappa)}{\bar{\sigma}_{0,T_0}}V_3^{\eps}(z)\left(\frac{1}{2} + \frac{\log(K/F_{0,T})}{\bar{\sigma}_{0,T_0}^2T_0}\right) \right.\\
&\left. + \frac{\lambda_0(T_0,T,\kappa)}{\bar{\sigma}_{0,T_0}}V_0^{\delta}(z) + \frac{\lambda_1(T_0,T,\kappa)}{\bar{\sigma}_{0,T_0}}V_0^{\delta}(z) \left(\frac{1}{2} + \frac{\log(K/F_{0,T})}{\bar{\sigma}_{0,T_0}^2T_0}\right)\right)\frac{\partial \bar{P}}{\partial \sigma}.
\end{align*}
Now, we convert the price $P^{\eps,\delta}$ to a Black implied volatility $I$:
$$C_B(T_0,K,I(T_0,K,T)) = P^{\eps,\delta} = \bar{P} + \bar{P}^{\eps,\delta} + \cdots.$$
\begin{obs*}
Since we do not have a spot price readily available for trade in our model we work with futures. Thus, differently from what is done in the equity case, we consider the Black implied volatility instead of the Black--Scholes implied volatility.
\end{obs*}
Then, expand $I(T_0,K,T)$ around $\bar{\sigma}_{0,T_0}$:
$$I(T_0,K,T) - \bar{\sigma}_{0,T_0} = \sqrt{\eps}I_{1,0}(T_0,K,T) + \sqrt{\delta}I_{0,1}(T_0,K,T) + \cdots.$$
Hence, matching both expansions gives us
\begin{align*}
\sqrt{\eps}I_{1,0}(T_0,K,T) &= \frac{3}{2}\frac{\lambda_3(T_0,T,\kappa)}{\lambda_{\sigma}(T_0,T,\kappa)}\frac{V_3^{\eps}(z)}{\bar{\eta}(z)} + \frac{\lambda_3(T_0,T,\kappa)}{\lambda_{\sigma}^3(T_0,T,\kappa)}\frac{V_3^{\eps}(z)}{\bar{\eta}^3(z)}\frac{\log(K/F_{0,T})}{T_0}, \\ \\
\sqrt{\delta}I_{1,0}(T_0,K,T) &= \left(\frac{\lambda_0(T_0,T,\kappa)}{\lambda_{\sigma}(T_0,T,\kappa)} + \frac{1}{2}\frac{\lambda_1(T_0,T,\kappa)}{\lambda_{\sigma}(T_0,T,\kappa)} \right) \frac{V_0^{\delta}(z)}{\bar{\eta}(z)} \\
&\hskip .5cm + \frac{\lambda_1(T_0,T,\kappa)}{\lambda_{\sigma}^3(T_0,T,\kappa)}\frac{V_0^{\delta}(z)}{\bar{\eta}^3(z)}\frac{\log(K/F_{0,T})}{T_0}.
\end{align*}
So, in terms of the reduced variable $\LMMR$, the log-moneyness to maturity ratio defined by
$$\LMMR = \frac{\log(K/F_{0,T})}{T_0},$$
the first-order approximation of the implied volatility $I(T_0,K,T)$ can be written as
\begin{align}
\hskip .55cm I(T_0,K,T) \approx\,\,& \bar{\eta}(z)\bar{b}(T_0,T,\kappa) +\frac{V_3^{\eps}(z)}{\bar{\eta}(z)}b^{\eps}(T_0,T,\kappa) + \frac{V_0^{\delta}(z)}{\bar{\eta}(z)} b^{\delta}(T_0,T,\kappa) \label{eq:approx_I}\\
&+ \left(\frac{V_3^{\eps}(z)}{\bar{\eta}^3(z)}a^{\eps}(T_0,T,\kappa) + \frac{V_0^{\delta}(z)}{\bar{\eta}^3(z)}a^{\delta}(T_0,T,\kappa)\right)\LMMR, \nonumber
\end{align}
where
\begin{align}
\bar{b}(T_0,T,\kappa) &= \lambda_{\sigma}(T_0,T,\kappa), \label{eq:b_bar} \\
b^{\eps}(T_0,T,\kappa) &= \frac{3}{2}\frac{\lambda_3(T_0,T,\kappa)}{\lambda_{\sigma}(T_0,T,\kappa)}, \label{eq:b_eps} \\
b^{\delta}(T_0,T,\kappa) &= \frac{\lambda_0(T_0,T,\kappa)}{\lambda_{\sigma}(T_0,T,\kappa)} + \frac{1}{2}\frac{\lambda_1(T_0,T,\kappa)}{\lambda_{\sigma}(T_0,T,\kappa)} , \label{eq:b_delta} \\
a^{\eps}(T_0,T,\kappa) &= \frac{\lambda_3(T_0,T,\kappa)}{\lambda_{\sigma}^3(T_0,T,\kappa)}, \label{eq:a_eps} \\
a^{\delta}(T_0,T,\kappa) &= \frac{\lambda_1(T_0,T,\kappa)}{\lambda_{\sigma}^3(T_0,T,\kappa)}. \label{eq:a_delta}
\end{align}
Therefore, the model predicts at first-order accuracy and for fixed maturities $T_0$ and $T$ that the implied volatility is affine in the $\LMMR$ variable.

\begin{obs*}
The importance of considering the derivative price as a function of the future price then unfolds into Equation \eqref{eq:D1Vega}. Indeed, if one considers $C_B$ as a function of the spot value $V$, as it is done in \cite{jaimungal_futures}, such formula would not be true. Hence, the calculations performed below could not be carried out following the standard steps presented in \cite{fouque_stochastic_vol_new}. In addition, the $\LMMR$ variable here must be defined with respect with the future price, providing again one more reason that the right variable to be considered is the future price instead of the spot price.
\end{obs*}

\subsection{Calibration Procedure}\label{sec:calibration_procedure}

Suppose that at the present time $t=0$, there is available the finite set of Black implied volatilities $\{I(T_{0ij}, K_{ijl}, T_i)\}$, which we understand in the following way: for each $i$ (i.e. for each future price $F_{0,T_i}$) there are available call option prices with maturities $T_{0ij}$ and, for each of these maturities, strikes $K_{ijl}$.

Since the data is more abundant in the $K$ direction, we will first linearly regress the implied volatilities against the variable $\LMMR_{ijl}$,
$$\LMMR_{ijl} = \frac{\log(K_{ijl}/F_{0,T_i})}{T_{0ij}},$$
for fixed $i$ and $j$. More precisely, we will use the least-squares criterion to perform this regression, namely
\begin{align}\label{eq:a_b_hat}
(\hat{a}_{ij}, \hat{b}_{ij}) = \ds \underset{(a_{ij},b_{ij})}{\operatorname{argmin}} \sum_l \left( I(T_{0ij}, K_{ijl}, T_i) - \left( a_{ij} \LMMR_{ijl} + b_{ij} \right) \right)^2.
\end{align}
Now, using Equation \eqref{eq:approx_I}, we first regress the estimate $\hat{a}$ against $a^{\eps}$ and $a^{\delta}$:
\begin{align}\label{eq:a_b_kappa_hat}
\hskip 1cm (\hat{a}_0, \hat{a}_1, \hat{\kappa})\hspace{-2pt} = \hspace{-2pt} \underset{(a_0,a_1,\kappa)}{\operatorname{argmin}} \hspace{-2pt} \sum_{i,j} \hspace{-2pt} \left( \hat{a}_{ij}\hspace{-2pt} - \hspace{-2pt}\left( a_0 a^{\eps}(T_{0ij},T_i,\kappa) \hspace{-2pt}+ \hspace{-2pt}a_1 a^{\delta}(T_{0ij},T_i,\kappa)\right) \right)^2
\end{align}
and then knowing $(\hat{a}_0, \hat{a}_1, \hat{\kappa})$ we regress $\hat{b}$ against $\bar{b}$, $b^{\eps}$ and $b^{\delta}$:
\begin{align}
\hat{b}_0 &= \underset{b_0}{\operatorname{argmin}} \sum_{i,j} \left( \hat{b}_{ij} - \left( b_0\bar{b}(T_{0ij},T_i,\hat{\kappa}) +  b_0^2\left( \hat{a}_0 b^{\eps}(T_{0ij},T_i,\hat{\kappa}) + \hat{a}_1 b^{\delta}(T_{0ij},T_i,\hat{\kappa})\right)\right) \right)^2. \label{eq:b_hat}
\end{align}

Therefore, we find the following estimates for the market group parameters:
\begin{align}
\widehat{\bar{\eta}(z)} &= \hat{b}_0, \label{eq:eta_bar_hat}\\
\widehat{V_3^{\eps}(z)} &= \hat{a}_0 \hat{b}_0^3, \label{eq:V3eps_hat} \\
\widehat{V_0^{\delta}(z)} &= \hat{a}_1 \hat{b}_0^3, \label{eq:V0delta_hat}
\end{align}
and $\hat{\kappa}$ is given by Equation \eqref{eq:a_b_kappa_hat}. In order to perform the above minimizations, the initial guesses are of utmost importance. Since we expect terms of order $\sqrt{\eps}$ to be small, we set the initial guesses of $a_0$ and $a_1$ to be 0. On the other hand, we can construct initial guesses of $\kappa$ and $b_0$ by estimating $\kappa$ and $\bar{\eta}(z)$ from historical data of $F_{t,T}$.

\subsection{Calibration Example}

In this section, we will exemplify the calibration procedure described in Section \ref{sec:calibration_procedure}. The goal of this section is merely the illustration of the calibration procedure.

It is important to notice that the calibration procedure requires only simple regressions. This is a huge improvement from the formulas derived in \cite{jaimungal_futures} that depend upon a more computationally demanding first-order approximation.

The data considered were Black implied volatilities of call and put options on the crude-oil future contracts on October 16th, 2013. On this day, 533 implied volatilities are available. This data is organized as follows: for each future contract (i.e. for each maturity $T_i$), there is one option maturity $T_{0ij}$ and 41 strikes $K_{ijl}$. By contractual specifications, the option maturity is roughly one month before the maturity of its underlying future contract (i.e. $T_i \approx T_{0ij} + 30$). The future prices are shown in Figure \ref{fig:futures} and since there is no clear seasonality component, we set $s \equiv 0$ in \eqref{eq:sde_risk_neutral}. The calibration of our model to all the available data is shown in Figure \ref{fig:iv_all} and Table \ref{tab:Calibration_all}. Since the implied volatilities curves present a noticeable smile for short maturities (30 and 60 days), we also calibrate our model to implied volatilities with maturity greater than 90 days. This is shown in Figure \ref{fig:iv} and Table \ref{tab:Calibration_90}. This is also consistent to our model, since the model requires enough time to maturity so the fast mean-reverting process $Y^{\eps}$ has enough time to oscillate around its ergodic mean. Furthermore, if one wants to capture the convexity of the implied volatility smile, one would have to use the second-order approximation as it is done for the Equity markets in \cite{secondorder}. In this numerical example, several reasonable initial guesses of $\kappa$ and $b_0$ where tested, all of them leading to the same calibrated parameters. Hence, the estimation of $\kappa$ and $b_0$ using historical data of $F_{t,T}$ was not necessary.

We show in Figures \ref{fig:iv} and \ref{fig:iv_all} the implied volatility fit for different maturities, where the solid line is the model implied volatility and the circles are the implied volatilities observed in the market. The shortest maturities implied volatility curves are on the leftmost thread and the maturity increases clockwise. The calibrated group market parameters are given in Tables \ref{tab:Calibration_90} and \ref{tab:Calibration_all}. It is important to notice that $V_3^{\eps}(z)$ and $V_0^{\delta}(z)$ are indeed small and hence these parameters are compatible with our model.

\begin{table}[h!]
\begin{center}
 \begin{tabular}{c c}
 \hline Parameter & Value  \\
 \hline
 $\hat{\kappa}$ & 0.1385 \\
 $\widehat{\bar{\eta}(z)}$ & 0.21967  \\
 $\widehat{V_3^{\eps}(z)}$ & -0.00017637 \\
 $\widehat{V_0^{\delta}(z)}$ &  -0.012656\\
 \hline
 \end{tabular}
\end{center}
\caption{Calibrated parameters using options with maturity greater than 90 days.} \label{tab:Calibration_90}
\begin{center}
 \begin{tabular}{c c}
 \hline Parameter & Value  \\
 \hline
 $\hat{\kappa}$ & 0.30853 \\
 $\widehat{\bar{\eta}(z)}$ & 0.23773 \\
 $\widehat{V_3^{\eps}(z)}$ & -0.00011823  \\
 $\widehat{V_0^{\delta}(z)}$ &  -0.007633\\
 \hline
 \end{tabular}
\end{center}
\caption{Calibrated Parameters using all data available} \label{tab:Calibration_all}
\end{table}
\begin{figure}[h!]
\begin{center}
\scalebox{0.6}{\includegraphics{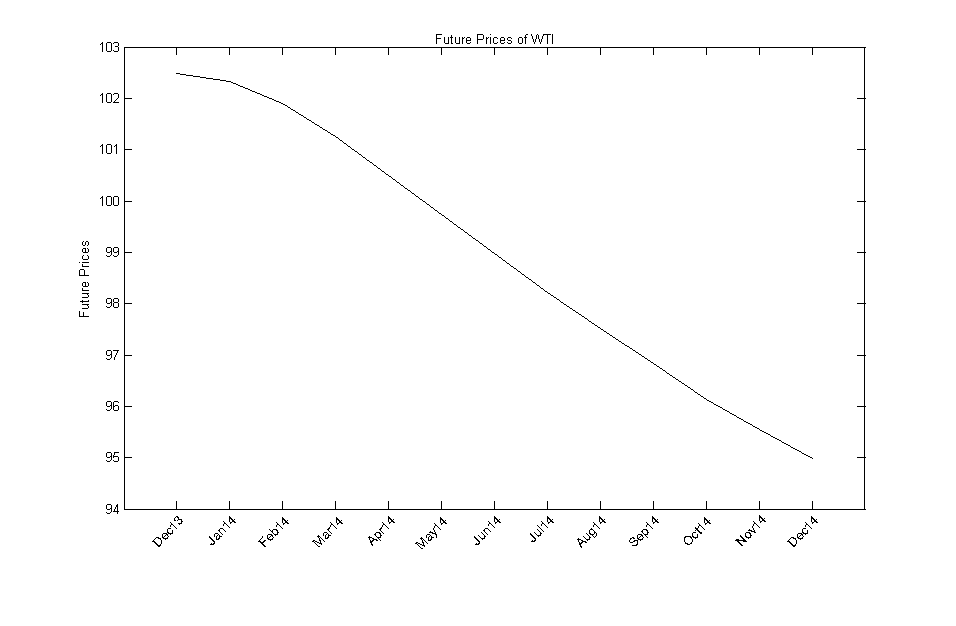}} \vspace{-25pt} \caption{Future prices on October 16th, 2013.} \label{fig:futures}
\end{center}
\end{figure}

\begin{figure}[h!]
\begin{center}
\scalebox{0.6}{\includegraphics{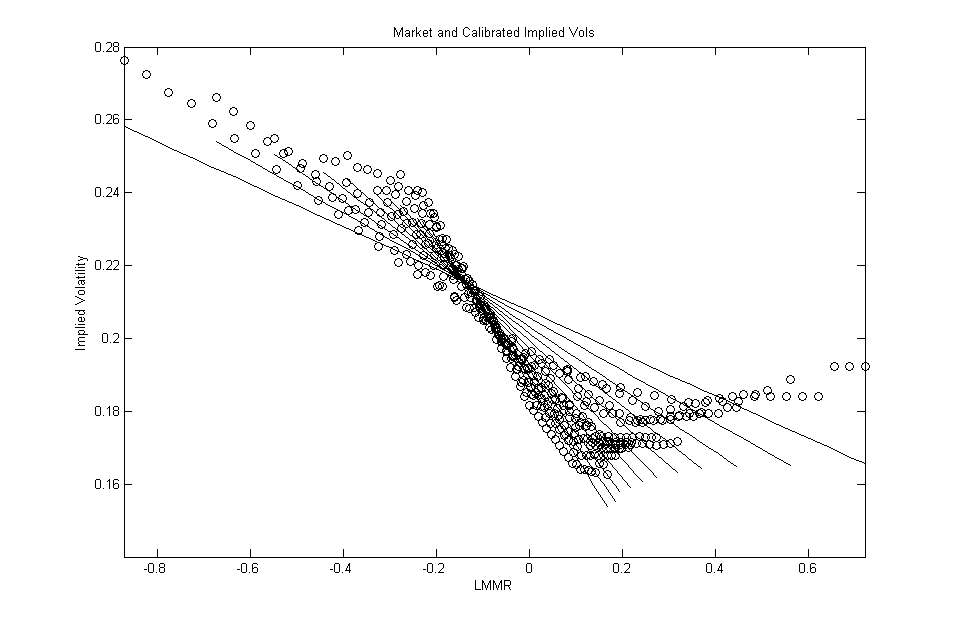}} \vspace{-25pt}  \caption{Market (circles) and calibrated (solid lines) implied volatilities for options on crude-oil futures with maturity greater than 90 days.} \label{fig:iv}  \vspace{25pt}
\scalebox{0.6}{\includegraphics{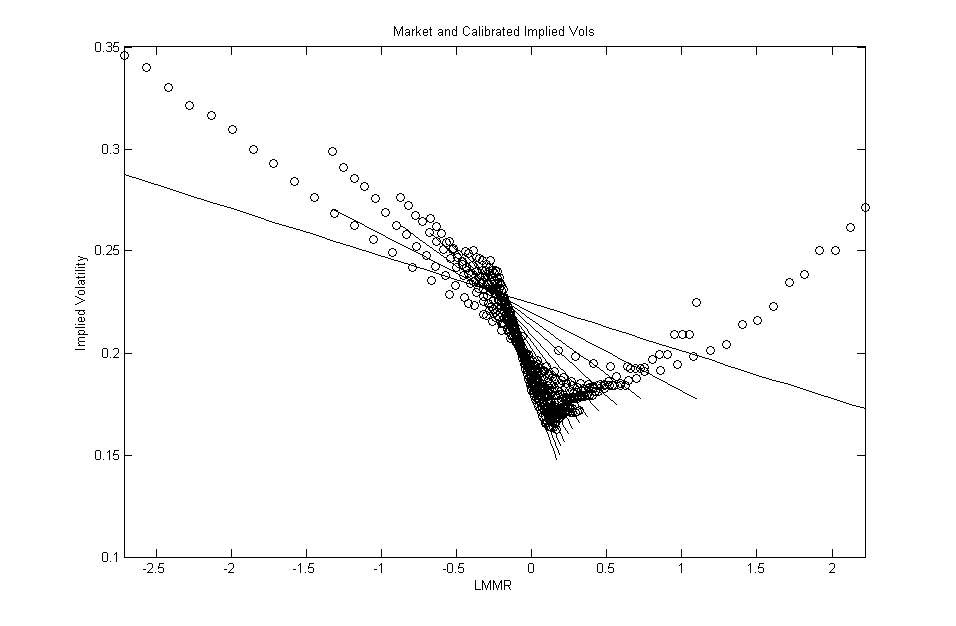}}\vspace{-10pt}   \caption{Market (circles) and calibrated (solid lines) implied volatilities for options on crude-oil futures using all data available.} \label{fig:iv_all}
\end{center}
\end{figure}

\clearpage

\section{Comparison to Previous Work}\label{sec:comparison}

Multiscale stochastic volatility models were considered in the context of mean reverting asset price in the papers \cite{jaimungal_futures} and \cite{ multiscale_aymptotic_mean_rever}. Both derived the first-order approximation of options on the spot price. However, only \cite{jaimungal_futures} handled options on futures. Therefore, we will focus this section on this work.

In \cite{jaimungal_futures}, the authors have applied the method first developed in \cite{multiscale_option_interest_rate} to approximate the price of options on commodity futures. This method consists of writing the payoff function as its first-order Taylor polynomial around the zero-order term of the future price expansion \eqref{eq:h_0} and then applying the singular-perturbation arguments introduced by Fouque \textit{et al.} to finally derive the approximation. Another aspect of their method that is fundamentally different of ours is that they consider the variable of interest to be the spot price of the commodity, denoted here by $V_t$, while we consider the future prices of $V$, denoted here by $F_{t,T}$. As we will see below, this is one of the reasons no simple calibration procedure can be designed using their first-order approximation.

The authors present two classes of models: one-factor and two-factor models. Both models share one common aspect: the commodity presents fast-mean reversion stochastic volatility. Our model \eqref{eq:sde_risk_neutral} is thus an extension of their one-factor model, where we have added the slow time-scale to the stochastic volatility dynamics. However, compared to \cite{jaimungal_futures}, the main contributions of our work are:
\begin{itemize}

\item[(a)] we do not rely on the Taylor expansion of the payoff function to derive the first-order approximation of options prices and therefore no additional smoothness assumption is necessary. The inversion argument presented in Section \ref{sec:perturbation} allows us to overcome this smoothness restriction to the payoff function;

\item[(b)] our first-order correction (see Summary \ref{sec:summary} for instance) is a substantial betterment of theirs, because it involves only Greeks of the zero-order term. Their first-order correction presents a complicated term that involves the expectation:
    $$\bE_{\bQ}[h_0(T_0, \overline{U}_{T_0})\varphi'(h_0(T_0, \overline{U}_{T_0}))\ | \ \overline{U}_t = u],$$
    where we are using the notation established in the previous sections and $\overline{U}$ is the process $U$ of Equation \eqref{eq:sde_risk_neutral} with $\eta(y,z) = \bar{\eta}(z)$;

\item[(c)] we present a simple calibration procedure of the market group parameters. The simple expression of our first-order correction is one of the reasons such calibration procedure is possible. However, the essential aspect of our method that is the cornerstone of the procedure is that we consider the future price $F_{t,T}$ as the variable, and not spot price $V_t$. So, since the future price is a martingale (as opposed to the $V$), better formulas for the Greeks of the 0-order term are available. Mainly, Equation \eqref{eq:D1Vega} holds true.

\end{itemize}

\section{Concluding Remarks and Future Directions}\label{sec:conclusion}

We have presented a general method to derive the first-order approximation of compound derivatives and developed it thoroughly in the case of derivatives on future contracts. Although the method may seem to be involved, it does not require any additional hypothesis on the regularity of the payoff function other than the ones inherent to the perturbation theory. Moreover, we presented a calibration procedure associated with the method, for which we derive formulas for the market group parameters. A practical numerical example of the calibration procedure is given using data of Black implied volatilities of options on crude-oil futures.

One direction for further study would be to connect the model and asymptotic expansions proposed in this work with the celebrated Schwartz-Smith \cite{schwartz-smith} and
Gibson-Schwartz \cite{gibson1990} models for commodity prices. In particular, an important question would be the computation of the risk premium as in \cite{gibson1990}.


\begin{appendix}

\section{PDE Expansion}\label{sec:appendix_pde_exp}

Formally write, for $k=1,2,3$,
$$\psi^{\eps,\delta}_k(t,x,y,z,T) = \sum_{i,j \geq 0} (\sqrt{\eps})^i(\sqrt{\delta})^j \psi_{k,i,j}(t,x,y,z,T).$$
In what follows we will compute only the terms of the above expansion that are necessary for the computation of the first-order approximation of derivatives on $F_{t,T}$.

\subsection{Expanding $\psi_1^{\eps,\delta}$}

By the chain rule
$$\psi^{\eps,\delta}_1(t,x,y,z,T) = \frac{1}{\ds\frac{\partial H}{\partial x}^{\eps,\delta}(t,x,y,z,T)},$$
and then one can easily see that
$$\psi_{1,0,0}(t,x,y,z,T) = \frac{1}{\ds\frac{\partial H_{0}}{\partial x}(t,x,z,T)} = \frac{\partial h_0}{\partial u}(t,H_0(t,x,z,T),z,T).$$
Since
$$\frac{\partial h_0}{\partial u}(t,u,z,T) = e^{-\kappa(T-t)}h_0(t,u,z,T),$$
we have the following formula for $\psi_{1,0,0}$:
\begin{align}\label{eq:psi100}
\psi_{1,0,0}(t,x,T) = e^{-\kappa(T-t)}h_0(t,H_0(t,x,z,T),z,T) = e^{-\kappa(T-t)}x.
\end{align}
Moreover, by Lemma \ref{lemma:inv},
\begin{align*}
H_{1,0}(t,x,z,T) &= -\frac{h_{1,0}(t,H_0(t,x,z,T),z,T)}{\ds\frac{\partial h_0}{\partial u}(t,H_0(t,x,z,T),z,T)} \\
&= -\frac{g(t,T)V_3(z)e^{-3\kappa(T-t)}x}{e^{-\kappa(T-t)}x} = -g(t,T)V_3(z)e^{-2\kappa(T-t)},
\end{align*}
which is independent of $x$, and therefore,
\begin{align}\label{eq:psi110}
\psi_{1,1,0}(t,x,z,T) = - \frac{1}{\ds\left(\frac{\partial H_0}{\partial x}(t,x,z,T)\right)^2} \frac{\partial H_{1,0}}{\partial x}(t,x,z,T) = 0.
\end{align}
We also have
\begin{align}\label{eq:psi101}
\psi_{1,0,1}(t,x,z,T) = 0.
\end{align}

\subsection{Expanding $\psi_2^{\eps,\delta}$}

Recall that the first four terms of the expansion of $h^{\eps,\delta}$ do not depend on $y$. Thus
\begin{align}\label{eq:psi20}
\psi_{2,0,0} = \psi_{2,0,1} = \psi_{2,1,0} = \psi_{2,1,1} = 0.
\end{align}
Furthermore, again by the chain rule, we have
$$\psi_2^{\eps,\delta}(t,x,y,z,T) = -\psi_1^{\eps,\delta}(t,x,y,z,T) \frac{\partial H}{\partial y}^{\eps,\delta}(t,x,y,z,T).$$
From this, we get
\begin{align*}
\psi_{2,2,0}(t,x,y,z,T) &= -\psi_{1,0,0}(t,x,T) \frac{\partial H_{2,0}}{\partial y}(t,x,y,z,T) \\
&= -e^{-\kappa(T-t)}x \frac{\partial H_{2,0}}{\partial y}(t,x,y,z,T).
\end{align*}
In order to compute $H_{2,0}$ we need to go further in the expansions of $h^{\eps,\delta}$ and $H^{\eps,\delta}$. One can compute the term of order $(2,0)$ in $H^{\eps,\delta}$ and then conclude
\begin{align*}
H_{2,0}(t,x,y,z,T) &= -\frac{h_{2,0}(t,H_0(t,x,z,T),y,z,T)}{\frac{\partial h_0}{\partial u}(t,H_0(t,x,z,T),z,T)} \\
&- \frac{1}{2} \frac{\frac{\partial^2 h_{0}}{\partial^2 u}(t,H_0(t,x,z,T),z,T)}{\frac{\partial h_0}{\partial u}(t,H_0(t,x,z,T),z,T)}H_{1,0}^2(t,x,z,T) \\
& - \frac{\frac{\partial h_{1,0}}{\partial u}(t,H_0(t,x,z,T),z,T)}{\frac{\partial h_0}{\partial u}(t,H_0(t,x,z,T),z,T)}H_{1,0}(t,x,z,T),
\end{align*}
which implies
$$\frac{\partial H_{2,0}}{\partial y}(t,x,y,z,T) = -\frac{\frac{\partial h_{2,0}}{\partial y}(t,H_0(t,x,z,T),y,z,T)}{\frac{\partial h_0}{\partial u}(t,H_0(t,x,z,T),z,T)}.$$
From \eqref{eq:h_20}, we know that
\begin{align*}
\frac{\partial h_{2,0}}{\partial y}(t,u,y,z,T) &= -\frac{1}{2} \frac{\partial \phi}{\partial y}(y,z) \frac{\partial^2 h_0}{\partial u^2}(t,u,z,T) \\
&= -\frac{1}{2}e^{-2\kappa(T-t)} \frac{\partial \phi}{\partial y}(y,z) h_0(t,u,z,T),
\end{align*}
and hence
$$\frac{\partial H_{2,0}}{\partial y}(t,x,y,z,T) = \frac{\frac{1}{2}e^{-2\kappa(T-t)} \frac{\partial \phi}{\partial y}(y,z)x}{e^{-\kappa(T-t)}x} = \frac{1}{2}e^{-\kappa(T-t)}\frac{\partial \phi}{\partial y}(y,z).$$
Finally,  this gives the formula
\begin{align}\label{eq:psi220}
\psi_{2,2,0}(t,x,y,z,T) = -\frac{1}{2}e^{-2\kappa(T-t)} \frac{\partial \phi}{\partial y}(y,z) x.
\end{align}

\subsection{Expanding $\psi_3^{\eps,\delta}$}

By the chain rule, we have
$$\psi_3^{\eps,\delta}(t,x,y,z,T) = -\psi_1^{\eps,\delta}(t,x,y,z,T) \frac{\partial H}{\partial z}^{\eps,\delta}(t,x,y,z,T),$$
and then
\begin{align}
\psi_{3,0,0}(t,x,z) &= -\psi_{1,0,0}(t,x,T) \frac{\partial H_{0}}{\partial z}(t,x,y,z,T) \label{eq:psi300}\\
&= \frac{\partial h_{0}}{\partial z}(t,H_0(t,x,z,T),z,T) \nonumber \\
&= \frac{1}{2\kappa} \bar{\eta}(z) \bar{\eta}'(z)(1 - e^{-2\kappa (T-t)}) x. \nonumber
\end{align}

\subsection{An Expansion of the Pricing PDE}

Define
$$\Psi_{n,p}^{k,m}(t,x,y,z,T) = \sum_{i=0}^n \sum_{j=0}^p\psi_{k,i,j}(t,x,y,z,T)\psi_{m,n-i,p-j}(t,x,y,z,T),$$
which is the $(n,p)$-order coefficient of the formal power series of $\psi^{\eps,\delta}_k\psi^{\eps,\delta}_m$:
\begin{align*}
&\psi^{\eps,\delta}_k(t,x,y,z,T)\psi^{\eps,\delta}_m(t,x,y,z,T) \\
&= \sum_{i,j,l,r \geq 0} (\sqrt{\eps})^{l+i}(\sqrt{\delta})^{j+r} \psi_{k,i,j}(t,x,y,z,T)\psi_{m,l,r}(t,x,y,z,T) \\
&=\sum_{n,p \geq 0} (\sqrt{\eps})^{n}(\sqrt{\delta})^{p} \left( \sum_{i=0}^n \sum_{j=0}^p\psi_{k,i,j}(t,x,y,z,T)\psi_{m,n-i,p-j}(t,x,y,z,T) \right).
\end{align*}

Therefore, we have the following formal expansion for the operator $\cL^{\eps,\delta}$ where we drop the variables $(t,x,y,z,T)$ for simplicity,

\begin{align*}
\cL^{\eps,\delta} = &\frac{1}{\eps} \left( \cL_0 + \frac{1}{2}\Psi^{2,2}_0 \beta^2(y) \frac{\partial^2}{\partial x^2} + \psi_{2,0,0} \beta^2(y) \frac{\partial^2}{\partial x \partial y} \right) \\
+& \frac{1}{\sqrt{\eps}}\left(\rhor_{1} \Psi^{1,2}_0 \eta(y,z)\beta(y) \frac{\partial^2}{\partial x^2} + \rhor_{1} \psi_{1,0,0} \eta(y,z)\beta(y) \frac{\partial^2}{\partial x \partial y} \right.\\
+& \left. \frac{1}{2}\Psi^{2,2}_{1,0} \beta^2(y) \frac{\partial^2}{\partial x^2} + \psi_{2,1,0} \beta^2(y) \frac{\partial^2}{\partial x \partial y}\right) \\
+& \left( \frac{\partial}{\partial t} + \frac{1}{2}\Psi^{1,1}_0 \eta^2(y,z) \frac{\partial^2}{\partial x^2} - r \cdot + \rhor_{1} \Psi^{1,2}_{1,0} \eta(y,z)\beta(y) \frac{\partial^2}{\partial x^2} \right.\\
+& \left. \rhor_{1} \psi_{1,1,0} \eta(y,z)\beta(y) \frac{\partial^2}{\partial x \partial y} + \frac{1}{2}\Psi^{2,2}_{2,0} \beta^2(y) \frac{\partial^2}{\partial x^2} + \psi_{2,2,0} \beta^2(y) \frac{\partial^2}{\partial x \partial y}\right) \\
+& \sqrt{\eps} \left(\frac{1}{2} \Psi^{2,2}_{3,0} \beta^2(y) \frac{\partial^2}{\partial x^2} + \psi_{2,3,0} \beta^2(y) \frac{\partial^2}{\partial x \partial y} + \rhor_{1} \Psi_{2,0}^{1,2} \eta(y,z) \beta(y) \frac{\partial^2}{\partial x^2}\right. \\
+& \left.\rhor_{1} \psi_{1,2,0} \eta(y,z) \beta(y) \frac{\partial^2}{\partial x \partial y} + \frac{1}{2} \Psi^{1,1}_{1,0} \beta^2(y) \frac{\partial^2}{\partial x^2} \right) \\
+& \sqrt{\delta}\left( \rhor_{2} \Psi^{1,3}_0 \eta(y,z)g(z) \frac{\partial^2}{\partial x^2} + \rhor_{2} \psi_{1,0,0} \eta(y,z)g(z) \frac{\partial^2}{\partial x \partial z} \right. \\
+& \frac{1}{2}\Psi^{2,2}_{2,1} \beta^2(y) \frac{\partial^2}{\partial x^2} + \psi_{2,2,1} \beta^2(y) \frac{\partial^2}{\partial x \partial y} \\
+& \left. \rhor_{2} \Psi^{1,2}_{1,1} \eta(y,z)\beta(y) \frac{\partial^2}{\partial x^2} + \rhor_{2} \psi_{1,1,1} \eta(y,z)\beta(y) \frac{\partial^2}{\partial x \partial y} + \frac{1}{2}\Psi^{1,1}_{0,1} \eta^2(y,z) \frac{\partial^2}{\partial x^2} \right) \\
+& \sqrt{\frac{\delta}{\eps}} \left(\rhor_{12} \Psi^{2,3}_0 \beta(y)g(z) \frac{\partial^2}{\partial x^2} + \rhor_{12} \psi_{3,0,0} \beta(y)g(z) \frac{\partial^2}{\partial x \partial y} \right.\\
+& \rhor_{12} \psi_{2,0,0} \beta(y)g(z) \frac{\partial^2}{\partial x \partial z} + \rhor_{12} \beta(y) g(z) \frac{\partial^2}{\partial y \partial z} + \frac{1}{2}\Psi^{2,2}_{1,1} \beta^2(y) \frac{\partial^2}{\partial x^2} \\
+& \left. \psi_{2,1,1} \beta^2(y) \frac{\partial^2}{\partial x \partial y} + \rhor_{1} \Psi^{1,2}_{0,1} \eta(y,z)\beta(y) \frac{\partial^2}{\partial x^2} + \rhor_{1} \psi_{1,0,1} \eta(y,z)\beta(y) \frac{\partial^2}{\partial x \partial y}\right). \\
+& \cL^{\eps,\delta}_R.
\end{align*}
Then we write, with obvious notation,
$$\cL^{\eps,\delta} = \frac{1}{\eps}\cD_{-2,0} + \frac{1}{\sqrt{\eps}}\cD_{-1,0} + \cD_{0,0} + \sqrt{\eps} \cD_{1,0} + \sqrt{\delta} \cD_{0,1} + \sqrt{\frac{\delta}{\eps}} \cD_{-1,1} + \cL^{\eps,\delta}_R.$$

In order to simplify the  operators $\cD_{i,j}$, note that $\Psi^{2,k}_0 = 0$, for $k=1,2,3$,
$$\Psi^{2,2}_{1,0} = \Psi^{2,2}_{0,1} = \Psi^{1,2}_{1,0} = \Psi^{2,2}_{2,0} = \Psi^{2,2}_{3,0} = \Psi_{2,1}^{2,2} = \Psi_{1,1}^{1,2} = 0,$$
and $\Psi^{1,1}_{1,0} = \Psi^{1,1}_{0,1} = 0$. Hence,
\begin{align*}
&\bullet\cD_{-2,0} = \cL_0 + \frac{1}{2}
\cancelto{\scriptstyle 0}{\Psi^{2,2}_0 } \beta^2(y) \frac{\partial^2}{\partial x^2} + \cancelto{\scriptstyle 0}{\psi_{2,0,0}^{} } \beta^2(y) \frac{\partial^2}{\partial x \partial y} = \cL_0 \\ \\
&\bullet\cD_{-1,0} = \rhor_{1} \cancelto{\scriptstyle 0}{\Psi^{1,2}_0 } \eta(y,z)\beta(y) \frac{\partial^2}{\partial x^2} + \rhor_{1} \psi_{1,0,0} \eta(y,z)\beta(y) \frac{\partial^2}{\partial x \partial y} \\
&+\frac{1}{2}\cancelto{\scriptstyle 0}{\Psi^{2,2}_{1,0} } \beta^2(y) \frac{\partial^2}{\partial x^2} + \cancelto{\scriptstyle 0}{\psi_{2,1,0}^{} } \beta^2(y) \frac{\partial^2}{\partial x \partial y} = \rhor_{1} \psi_{1,0,0} \eta(y,z)\beta(y) \frac{\partial^2}{\partial x \partial y}\\ \\
&\bullet\cD_{0,0} = \frac{\partial}{\partial t} + \frac{1}{2}\Psi^{1,1}_0 \eta^2(y,z) \frac{\partial^2}{\partial x^2} - r\cdot + \rhor_{1} \cancelto{\scriptstyle 0}{\Psi^{1,2}_{1,0} } \eta(y,z)\beta(y) \frac{\partial^2}{\partial x^2} \\
&+ \rhor_{1} \cancelto{\scriptstyle 0}{\psi_{1,1,0}^{} } \eta(y,z)\beta(y) \frac{\partial^2}{\partial x \partial y} +
+ \frac{1}{2}\cancelto{\scriptstyle 0}{\Psi^{2,2}_{2,0} } \beta^2(y) \frac{\partial^2}{\partial x^2} + \psi_{2,2,0} \beta^2(y) \frac{\partial^2}{\partial x \partial y} \\
&= \frac{\partial}{\partial t} + \frac{1}{2} \psi_{1,0,0}^2\eta^2(y,z) \frac{\partial^2}{\partial x^2} - r\cdot + \psi_{2,2,0} \beta^2(y) \frac{\partial^2}{\partial x \partial y} \\ \\
&\bullet\cD_{1,0} = \frac{1}{2} \cancelto{\scriptstyle 0}{\Psi^{2,2}_{3,0} } \beta^2(y) \frac{\partial^2}{\partial x^2} + \psi_{2,3,0} \beta^2(y) \frac{\partial^2}{\partial x \partial y} + \rhor_{1} \cancelto{\scriptstyle\psi_{1,0,0}\psi_{2,2,0}}{\Psi_{2,0}^{1,2} } \eta(y,z) \beta(y) \frac{\partial^2}{\partial x^2} \\
&+ \rhor_{1} \psi_{1,2,0} \eta(y,z) \beta(y) \frac{\partial^2}{\partial x \partial y} + \frac{1}{2} \cancelto{\scriptstyle 0}{\Psi^{1,1}_{1,0} } \beta^2(y) \frac{\partial^2}{\partial x^2} \\
&= \left(\psi_{2,3,0} \beta^2(y) + \rhor_{1} \psi_{1,2,0} \eta(y,z) \beta(y)\right) \frac{\partial^2}{\partial x \partial y}
+ \rhor_{1}\psi_{1,0,0}\psi_{2,2,0} \eta(y,z) \beta(y) \frac{\partial^2}{\partial x^2} \\ \\
&\bullet\cD_{0,1} = \rhor_{2} \Psi^{1,3}_0 \eta(y,z)g(z) \frac{\partial^2}{\partial x^2} + \rhor_{2} \psi_{1,0,0} \eta(y,z)g(z) \frac{\partial^2}{\partial x \partial z} + \frac{1}{2} \cancelto{\scriptstyle 0}{\Psi^{2,2}_{2,1} } \beta^2(y) \frac{\partial^2}{\partial x^2} \\
&+ \psi_{2,2,1} \beta^2(y) \frac{\partial^2}{\partial x \partial y} + \rhor_{2} \cancelto{\scriptstyle 0}{\Psi^{1,2}_{1,1} } \eta(y,z)\beta(y) \frac{\partial^2}{\partial x^2} + \rhor_{2} \psi_{1,1,1} \eta(y,z)\beta(y) \frac{\partial^2}{\partial x \partial y} \\
&+ \frac{1}{2} \cancelto{\scriptstyle 0}{\Psi^{1,1}_{0,1} }\eta^2(y,z) \frac{\partial^2}{\partial x^2} = (\rhor_{2} \psi_{1,0,0} \psi_{3,0,0} \eta(y,z)g(z))\frac{\partial^2}{\partial x^2} \\
&+\rhor_{2} \psi_{1,0,0} \eta(y,z)g(z) \frac{\partial^2}{\partial x \partial z} + (\psi_{2,2,1} \beta^2(y) + \rhor_{2} \psi_{1,1,1} \eta(y,z)\beta(y)) \frac{\partial^2}{\partial x \partial y} \\ \\
&\bullet\cD_{-1,1} = \rhor_{12} \cancelto{\scriptstyle 0}{\Psi^{2,3}_0 } \beta(y)g(z) \frac{\partial^2}{\partial x^2} + \rhor_{23} \psi_{3,0,0} \beta(y)g(z) \frac{\partial^2}{\partial x \partial y} \\
&+ \rhor_{12} \cancelto{\scriptstyle 0}{\psi_{2,0,0}^{} } \beta(y)g(z) \frac{\partial^2}{\partial x \partial z} + \rhor_{12} \beta(y) g(z) \frac{\partial^2}{\partial y \partial z} + \frac{1}{2}\cancelto{\scriptstyle 0}{\Psi^{2,2}_{1,1} } \beta^2(y) \frac{\partial^2}{\partial x^2} \\
&+ \cancelto{\scriptstyle 0}{\psi_{2,1,1}^{} } \beta^2(y) \frac{\partial^2}{\partial x \partial y} + \rhor_{1} \cancelto{\scriptstyle 0}{\Psi^{1,2}_{0,1} } \eta(y,z)\beta(y) \frac{\partial^2}{\partial x^2} \\
&+ \rhor_{1} \cancelto{\scriptstyle 0}{\psi_{1,0,1}^{} } \eta(y,z)\beta(y) \frac{\partial^2}{\partial x \partial y} = \rhor_{12} \psi_{3,0,0} \beta(y)g(z) \frac{\partial^2}{\partial x \partial y} + \rhor_{12} \beta(y) g(z) \frac{\partial^2}{\partial y \partial z}
\end{align*}

It is clear from the above choices that the coefficients of the operator $\cL^{\eps,\delta}_R$ are of order $(\eps+\delta)$. Hence, if $f$ is smooth and bounded with all derivatives bounded,
\begin{align}\label{eq:residual_PDE}
\cL^{\eps,\delta}_R f(t,x,y,z) = O(\eps+\delta).
\end{align}

\section{Proof of Theorem \ref{thm:accuracy}}\label{sec:proof}

By the regularization argument presented in \cite{fouque_call}, extended for the addition of the slow factor \cite{fouque_stochastic_vol_new}, we may assume that the payoff $\varphi$ is smooth and bounded with bounded derivatives. Additionally, in the reference mentioned above, the authors outlined an argument to improve the error bound from $O(\eps \log |\eps| + \delta)$ to $O(\eps + \delta)$.

Following the proof of Theorem 4.10 given in \cite{fouque_stochastic_vol_new} for the Equity case, we go further in the approximation of $P^{\eps,\delta}$ and define the higher-order approximation:
\begin{align}\label{eq:Phat}
\widehat{P}^{\eps,\delta} = \widetilde{P}^{\eps,\delta} + \eps(P_{2,0} + \sqrt{\eps}P_{3,0}) + \sqrt{\delta}(\sqrt{\eps}P_{1,1} + \eps P_{2,1}),
\end{align}
and moreover, we introduce
\begin{align}\label{eq:cL_hat}
\widehat{\cL}^{\eps,\delta} = \frac{1}{\eps} \cL_0 + \frac{1}{\sqrt{\eps}} \cL_1 + \cL_2 + \sqrt{\eps} \cL_3 + \sqrt{\delta} \cM_1 + \sqrt{\frac{\delta}{\eps}} \cM_3.
\end{align}
Necessary properties of the additional terms in the expansion \eqref{eq:Phat} can be derived in the same way as it is done in \cite{fouque_stochastic_vol_new} and thus we skip the details here. Next, we define the residual
\begin{align}
R^{\eps,\delta} = P^{\eps,\delta} - \widehat{P}^{\eps,\delta}
\end{align}
and use the pricing PDE \eqref{eq:pricing_PDE} to conclude that
\begin{align*}
\cL^{\eps,\delta} R^{\eps,\delta} &= \cL^{\eps,\delta}(P^{\eps,\delta} - \widehat{P}^{\eps,\delta}) = -\cL^{\eps,\delta}\widehat{P}^{\eps,\delta} \\
&= -( \widehat{\cL}^{\eps,\delta} + \cL^{\eps,\delta}_R )\widehat{P}^{\eps,\delta} = -\widehat{\cL}^{\eps,\delta} \widehat{P}^{\eps,\delta} - \cL^{\eps,\delta}_R \widehat{P}^{\eps,\delta}.
\end{align*}
Hence, mimicking the computations from \cite{fouque_stochastic_vol_new}, we can write
$$\widehat{\cL}^{\eps,\delta} \widehat{P}^{\eps,\delta} = \eps R_1^{\eps} + \sqrt{\eps\delta}R_2^{\eps} + \delta R_3^{\eps},$$
where $R_1^{\eps}$, $R_2^{\eps}$ and $R_3^{\eps}$ can be exactly computed as linear combinations of $\cL_k P_{i,j}$ and $\cM_l P_{i,j}$, for some $k,l,i,j$. However, the important fact about $R_i^{\eps}$ is that they are smooth functions of $t,x,y,z$, for small $\eps$ and $\delta$, bounded by smooth functions of $t,x,y,z$ independent of $\eps$ and $\delta$, uniformly bounded in $t,x,z$ and at most linearly growing in $y$. Additionally, from \eqref{eq:residual_PDE}, we can define $R_4^{\eps}$ and conclude that
$$R_4^{\eps} = \cL^{\eps,\delta}_R \widehat{P}^{\eps,\delta} = O(\eps+\delta),$$
since $\widehat{P}^{\eps,\delta}$ is bounded and smooth with all derivatives bounded. This follows from the same arguments presented in proof of Theorem 4.10 of \cite{fouque_stochastic_vol_new} with some minor differences to include $\cL_3$. Therefore, the residual $R^{\eps,\delta}$ solves the PDE
\begin{align}\label{eq:pde_R_eps_delta}
\left\{
\begin{array}{l}
 \cL^{\eps,\delta} R^{\eps,\delta} + \eps R_1^{\eps} + \sqrt{\eps\delta}R_2^{\eps} + \delta R_3^{\eps} + R_4^{\eps} = 0, \\ \\
 R^{\eps,\delta}(T_0,x,y,z,T) = -\eps(P_{2,0} + \sqrt{\eps}P_{3,0})(T_0,x,y,z,T) \\
 \hskip 2.8cm -\sqrt{\eps\delta}(P_{1,1} + \sqrt{\eps}P_{2,1})(T_0,x,y,z,T) ,
\end{array}
\right.
\end{align}
and then all the computations regarding the Feynman-Kac probabilistic representation of $R^{\eps,\delta}$ and the growth control of the source and final condition can be carried out just as in \cite{fouque_stochastic_vol_new} in order to conclude that $R^{\eps,\delta} = O(\eps+\delta)$. Lastly, the desired result follows because
$$|P^{\eps,\delta} - \widetilde{P}^{\eps,\delta}| \leq |R^{\eps,\delta}| + |\widehat{P}^{\eps,\delta} - \widetilde{P}^{\eps,\delta}|.$$ \qed



\section{Computing $P_{0,1}^{\delta}$ explicitly for European derivatives}\label{sec:p_1_delta_comp}

Let us write $P^{\delta}_{0,1}$ in terms of Greeks of $P_B$ by solving the PDE \eqref{eq:pde_p_1_delta}. Note first that
$$\frac{\partial P_0}{\partial z}(t,x,z,T) = \frac{\partial P_B}{\partial \sigma}(t,x, \bar{\sigma}_{t,T_0}(z,T))\frac{\partial \bar{\sigma}_{t,T_0}}{\partial z}(z,T).$$
Since $P_B$ satisfies the Black equation and the European derivative has maturity $T_0$, we have the relation between the Vega and the Gamma:
$$\frac{\partial P_B}{\partial \sigma}(t,x, \sigma) = (T_0-t) \sigma D_2P_B(t,x,\sigma).$$
Moreover,
$$\frac{\partial \bar{\sigma}_{t,T_0}}{\partial z}(z,T) = \bar{\eta}'(z) \lambda_{\sigma}(t,T_0,T,\kappa).$$
Combining the above equations, we obtain
\begin{align*}
&-f_0(t,T)\cA_0^{\delta}P_0 - f_1(t,T)\cA_1^{\delta}P_0 \\
&= -f_0(t,T)\cA_0^{\delta}P_0
-f_1(t,T)V_1^{\delta}(z) (T_0-t) \bar{\eta}(z) \bar{\eta}'(z) \lambda^2_{\sigma}(t,T_0,T,\kappa) D_1 D_2P_0 \\
&= -f_0(t,T)\cA_0^{\delta}P_0 - \tilde{f}_1(t,T)\cA_0^{\delta}D_1P_0,
\end{align*}
where
$$\tilde{f}_1(t,T) = f_1(t,T) (e^{-2\kappa (T-T_0)} - e^{-2\kappa (T-t)}).$$
Hence, we deduce the following PDE for $P_{0,1}^{\delta}$:
\begin{align}\label{eq:pde_p_1_delta_after}
\left\{
\begin{array}{l}
 \cL_B(\bar{\sigma}(t,z,T))P^{\delta}_{0,1} = -f_0(t,T)\cA_0^{\delta} P_0 - \tilde{f}_1(t,T)\cA_0^{\delta} D_1 P_0, \\ \\
 P^{\delta}_{0,1}(T_0,x,z,T) = 0.
\end{array}
\right.
\end{align}

By the same argument used to deduce Equation~\eqref{eq:P10_eps}, and using the linearity of the differential operators involved, we can easily conclude
$$P^{\delta}_{0,1}(t,x,z,T) = \left( \int_t^{T_0} f_0(u,T)du\right)\cA_0^{\delta} P_0 + \left(\int_t^{T_0} \tilde{f}_1(u,T)du\right)\cA_0^{\delta} D_1P_0.$$
Computing these integrals gives us the formulas
\begin{align}
&\lambda_0(t,T_0,T,\kappa) = \frac{1}{T_0-t}\int_t^{T_0} f_0(u,T) du \label{eq:lambda_0_comp}\\
&= \frac{e^{-\kappa(T-T_0)} - e^{-\kappa (T-t)}}{\kappa(T_0-t)} - \frac{e^{-3\kappa(T-T_0)} - e^{-3\kappa(T-t)}}{3\kappa(T_0-t)} \nonumber\\
&= \lambda(t,T_0,T,\kappa) -\lambda(t,T_0,T,3\kappa),\nonumber \\
&\lambda_1(t,T_0,T,\kappa) = \frac{1}{T_0-t} \int_t^{T_0} \tilde{f}_1(u,T) du \label{eq:lambda_1_comp}\\
&= e^{-2\kappa(T-T_0)} \frac{e^{-\kappa(T-T_0)} - e^{-\kappa(T-t)}}{\kappa(T_0-t)} - \frac{e^{-3\kappa(T-T_0)} - e^{-3\kappa(T-t)}}{3\kappa(T_0-t)} \nonumber\\
&= e^{-2\kappa(T-T_0)}\lambda(t,T_0,T,\kappa) -\lambda(t,T_0,T,3\kappa). \nonumber
\end{align}

\end{appendix}

\bibliography{my_bib_tex}
\bibliographystyle{plainnat}

\end{document}